\documentclass[journal]{IEEEtran}
\usepackage{ulem}
\usepackage{graphicx}
\graphicspath{ {./Figures/} }  
\usepackage{amsmath}
\usepackage[cmintegrals]{newtxmath}
\usepackage{bm,nicefrac}
\usepackage{url}
\usepackage{caption}
\usepackage{bm,nicefrac}
\usepackage{booktabs}

\usepackage{multicol,lipsum}
\usepackage{graphicx}
\usepackage{amsfonts}
\usepackage{subcaption}
\usepackage{amsmath}
\usepackage{lipsum}
\usepackage{stfloats}
\usepackage{stackengine}
\usepackage{float}   
\usepackage{cite}
\usepackage{epstopdf}
\usepackage{csquotes}
\usepackage{relsize}
\usepackage{xcolor}
\usepackage{balance}
\usepackage{enumitem}
\newcommand\mydots{\hbox to 1em{.\hss.\hss.}}
\usepackage[linesnumbered, ruled, vlined]{algorithm2e}
\usepackage{algpseudocode}
\usepackage{hyperref}
\usepackage{scalerel}
\usepackage{pifont}
\newcommand{\xmark}{\ding{55}}%
\usepackage{tikz}
\usepackage[font=small,labelfont=small]{caption}

\newtheorem{proposition}{Proposition}

\usepackage{authblk}
\allowdisplaybreaks
\usepackage{graphicx}
\usepackage{mathrsfs}
\usepackage{caption}

\usepackage[font=footnotesize]{caption}
\usepackage{stackengine}

\makeatletter
\newcommand\makebig[2]{%
\@xp\newcommand\@xp*\csname#1\endcsname{\bBigg@{#2}}%
\@xp\newcommand\@xp*\csname#1l\endcsname{\@xp\mathopen\csname#1\endcsname}%
\@xp\newcommand\@xp*\csname#1r\endcsname{\@xp\mathclose\csname#1\endcsname}%
}
\makeatother
\makebig{biggg} {3.0}
\makebig{Biggg} {3.5}
\makebig{bigggg}{4.0}
\makebig{Bigggg}{6.5}
\newcommand{\bieeeeq}{\begin{IEEEeqnarray}{rcl}}
	\newcommand{\eieeeeq}{\end{IEEEeqnarray}}
\newcommand{\eqtext}[1]{\ensuremath{\stackrel{\text{#1}}{=}}}  
\newcommand{\leqtext}[1]{\ensuremath{\stackrel{\text{#1}}{\leq}}}  

\algnewcommand{\Inputs}[1]{%
  \State \textbf{Inputs:}
  \Statex \hspace*{\algorithmicindent}\parbox[t]{.8\linewidth}{\raggedright #1}
}
\algnewcommand{\Initialize}[1]{%
  \State \textbf{Initialize:}
  \Statex \hspace*{\algorithmicindent}\parbox[t]{.8\linewidth}{\raggedright #1}
}
\algblock{Input}{EndInput}
\algnotext{EndInput}
\algblock{Output}{EndOutput}
\algnotext{EndOutput}

\usepackage{xcolor}
\def\BibTeX{{\rm B\kern-.05em{\sc i\kern-.025em b}\kern-.08em
  T\kern-.1667em\lower.7ex\hbox{E}\kern-.125emX}}

\def \AAb {\mathbf{A}_{\text{b}}}
\def \AAe {\mathbf{A}_{\text{e}}}
\def \BB {\textbf{B}}
\def \Rb {R_{\text{b}}}
\def \Re {R_{\text{e}}}
\def \Csec {C_{\text{sec}}}

\def \EE {\mathbf{E}}
\def \ee {\mathbf{e}}
\def \FFb {\mathbf{F}_{\text{b}}}
\def \FFe {\mathbf{F}_{\text{e}}}
\def \GG {\mathbf{G}}
\def \Hab {\mathbf{H}_{\text{ab}}}
\def \Har {\mathbf{H}_{\text{ar}}}
\def \Hrb {\mathbf{H}_{\text{rb}}}
\def \Hae {\mathbf{H}_{\text{ae}}}
\def \Hre {\mathbf{H}_{\text{re}}}
\def \Hhatb {\hat{\mathbf{H}}_{\text{b}}}
\def \Hhate {\hat{\mathbf{H}}_{\text{e}}}

\def \II {\mathbf{I}}

\def \Na {N_{\text{a}}}
\def \Nb {N_{\text{b}}}

\def \nb {\mathbf{n}_{\text{b}}}
\def \Ns {N_{\text{s}}}
\def \Ne {N_{\text{e}}}
\def \Nx {N_{\text{x}}}
\def \ne {\mathbf{n}_{\text{e}}}
\def \Pdiff {P_{\text{diff}}}

\def \ss {\mathbf{s}}
\def \TT {\mathbf{T}}
\def \That {\hat{\mathbf{T}}}
\def \Topt {\mathbf{T}^{\text{(opt)}}}
\def \UUb {\mathbf{U}_{\text{b}}}
\def \UUe {\mathbf{U}_{\text{e}}}

\def \XX {\mathbf{X}}
\def \yb {\mathbf{y}_{\text{b}}}
\def \ye {\mathbf{y}_{\text{e}}}

\def \betamin {\beta_{\text{min}}}
\def \betaa {\boldsymbol{\beta}}
\def \thetamin {\theta_{\text{min}}}
\def \thetamax {\theta_{\text{max}}}
\def \thetamaxstep {\theta_{\text{max,step}}}
\def \thetaa {\boldsymbol{\theta}}
\def \thetaahat {\hat{\boldsymbol{\theta}}}
\def \thetac {\theta_{\text{c}}}
\def \thetaaopt {\boldsymbol{\theta}^{\text{(opt)}}}
\def \Phii {\boldsymbol{\Phi}}
\def \phii {\boldsymbol{\phi}}
\def \sigmab {\sigma_{\text{b}}^2}
\def \sigmae {\sigma_{\text{e}}^2}
\def \lambdasin {\lambda_{\text{sin}}}
\def \lambdacos {\lambda_{\text{cos}}}
\def \lambdapow {\lambda_{\text{pow}}}
\def \lambdadiv {\lambda_{\text{div}}}
\def \lambdalogtwo {\lambda_{\text{log2}}}
\def \lambdasqrt {\lambda_{\text{sqrt}}}

\def \trace {\text{Tr}}
\def \tr {\text{tr}}

\def \diag {\text{diag}}
\def \Ree {\mathfrak{R}}

\usepackage{tikz,xcolor,hyperref}
\definecolor{lime}{HTML}{A6CE39}
\DeclareRobustCommand{\orcidicon}{%
	\begin{tikzpicture}
	\draw[lime, fill=lime] (0,0) 
	circle [radius=0.16] 
	node[white] {{\fontfamily{qag}\selectfont \tiny ID}};
	\draw[white, fill=white] (-0.0625,0.095) 
	circle [radius=0.007];
	\end{tikzpicture}
	\hspace{-2mm}
}

\foreach \x in {A, ..., Z}{%
	\expandafter\xdef\csname orcid\x\endcsname{\noexpand\href{https://orcid.org/\csname orcidauthor\x\endcsname}{\noexpand\orcidicon}}
}

\begin{document}

\title{Secrecy Rate Maximization in RIS-Assisted MIMO Systems Using a Practical Hardware Model}

\author{Rakesh Ranjan \orcidA{},\IEEEmembership{ Member,~IEEE}, Ahmad Sirojuddin \orcidB{}, \IEEEmembership{ Member,~IEEE}, Manjesh K. Hanawal\orcidC{},\IEEEmembership{ Senior Member,~IEEE}, Himanshu B. Mishra \orcidD{},\IEEEmembership{ Senior Member,~IEEE}, Wan-Jen Huang \orcidE{}, \IEEEmembership{Member,~IEEE}

\thanks{

R. Ranjan, and M. K. Hanawal are with the Department of Industrial Engineering and Operations Research, Indian Institute of Technology, Bombay, Maharashtra 400076, India (e-mail: rakeshranjan2911@gmail.com, mhanawal@iitb.ac.in)

A. Sirojuddin is with the Department of Electrical Engineering, Institut Teknologi Sepuluh Nopember, Surabaya 6011, Indonesia (e-mail: sirojuddin@its.ac.id)

H. B. Mishra is with the Department of Electronics Engineering, Indian Institute of Technology (Indian School of Mines), Dhanbad,  Jharkhand 826004, India (e-mail: himanshu@iitism.ac.in).

Wan-Jen Huang is with the Institute of Communications Engineering, National Sun Yat-sen University, Kaohsiung City 80424, Taiwan (e-mail: wjhuang@faculty.nsysu.edu.tw)
}
}

                           
\maketitle
\begin{abstract}

This study investigates a robust reconfigurable intelligent surface (RIS)-assisted multiple-input multiple-output (MIMO) system for secure wireless communication, in which a multi-antenna transmitter (Alice) sends confidential messages to a multi-antenna receiver (Bob) in the presence of an eavesdropper (Eve). Unlike idealized models, the reflecting elements (REs) of the RIS are assumed to possess inherent electrical resistance, introducing a practical non-ideal effect often neglected in prior research. The aim of the study is to maximize the secrecy rate of the MIMO system under perfect knowledge of the channel state information (CSI). To achieve this, the secrecy rate maximization problem is formulated and solved using a low-complexity joint optimization framework based on an adaptive projected gradient method (PGM), which simultaneously updates both the transmit precoding matrix and the RIS phase shifts. Solving the exact problem is computationally complex. Thus, a simplified variant is further introduced that maximizes the channel power difference rather than the exact secrecy rate. The simulation results show that this approximation yields a secrecy rate close to the true optimum while significantly reducing the computational cost. In addition, the proposed PGM with an adaptive step size initialization and control mechanism substantially improves the secrecy rate and reduces the computational time compared to the conventional fixed step size PGM. Overall, the simulation results confirm the effectiveness of the proposed PGM and demonstrate that adopting a practical RIS model is essential for establishing secure RIS-assisted MIMO communication links, especially under varying RE resistance values. 
  
\end{abstract}
\begin{IEEEkeywords}
Secrecy rate maximization, channel power difference maximization, MIMO, practical model reconfigurable intelligent surface, physical layer security.
\end{IEEEkeywords}

\section{Introduction}

\IEEEPARstart {W}{IRELESS} security is a crucial aspect of fifth-generation ($5$G) or beyond $5$G (B$5$G) networks. Traditionally, security in wireless transmission has relied on cryptographic techniques. However, these methods often introduce significant overheads and face challenges related to key distribution and management, which can compromise their reliability due to the inherent broadcast nature of the wireless medium~\cite{DBLP:journals/tvt/CumananDSTL14}. Consequently, secure transmission approaches that exploit the physical characteristics of wireless channels have gained increasing attention. Unlike traditional cryptographic methods, physical layer security (PLS) methods offer notable advantages such as lower complexity and the potential for keyless secure communication \cite{DBLP:journals/bstj/Shannon49a}. The performance of PLS methods is typically evaluated using the secrecy rate, which is quantified as the highest achievable rate at which the transmitter can securely and reliably communicate with the intended user (Bob) while preventing an eavesdropper (Eve) from intercepting the transmitted information~\cite{DBLP:journals/bstj/Wyner75a}. Multi-antenna techniques have been extensively employed to further improve the security of transmissions by leveraging the availability of additional spatial degrees of freedom. The secrecy rate analysis of multiple-input multiple-output (MIMO) channels in the presence of Eve has been extensively investigated in~\cite{DBLP:journals/tcom/LoykaC15,DBLP:journals/tcom/MukherjeeOT21,DBLP:journals/tsp/HanifTJG14}, with the findings providing valuable insights into the potential of spatial processing to enhance PLS.

Reconfigurable intelligent surface (RIS) has emerged as a promising technology to enhance both the spectral efficiency and PLS in B$5$G wireless communication networks. An RIS consists of a large array of passive reflecting elements (REs) designed as an impedance-based meta-surface embedded with reconfigurable electromagnetic components. These components allow each RE to reflect the incident waves with a precisely controllable phase shift, thereby enabling the reflected signal to be tailored as required to obtain the desired channel response\cite{DBLP:journals/comsur/GongLHNSKL20,DBLP:journals/tcom/WuZZYZ21,DBLP:journals/cm/PanRWKECRHWSYH21,DBLP:journals/ojcs/GuoLMALAW24}. For instance, by intelligently adjusting the phase shifts, the RIS can enhance the signal quality at the intended user (IU) while simultaneously reducing the signal strength at the non-intended user (NU). 
As a result, RIS plays a crucial role in enhancing the PLS of  wireless communication systems~\cite{DBLP:IRS-application1,DBLP:IRS-application2, DBLP:IRS_application3}. 
Recent studies have therefore focused on optimizing RIS configurations to maximize the secrecy rate performance using various design and optimization strategies~\cite{DBLP:journals/jsac/YuXSNS20,DBLP:journals/ojcs/AlmohamadTAFKHA20, DBLP:journals/tifs/ZhangDSAN21, DBLP:journals/corr/abs-2103-16696}. Collectively, these efforts highlight the central role of RIS in realizing secure and efficient wireless links through environment-adaptive signal control.

\subsection{Related Works}

This section provides a brief overview of related studies focused on improving the achievable rate and PLS in RIS-assisted communication systems. A summary of the key contributions of these studies is presented in Table~\ref{tab:related_work}. The authors in~\cite{DBLP:journals/wcl/SirojuddinPH22,DBLP:journals/wcl/NingCCF20} maximized the achievable rate in RIS-aided MIMO systems by obtaining the optimal values of the RIS phase shifts using either the dimension-wise sinusoidal maximization (DSM)~\cite{DBLP:journals/wcl/SirojuddinPH22} or the sum-path gain maximization (SPGM)~\cite{DBLP:journals/wcl/NingCCF20} method. Rajan et al. \cite{DBLP:journals/icl/RanjanBMM23} analyzed the achievable rate maximization problem in RIS-aided orthogonal frequency division multiplexing (OFDM) systems and proposed a gradient ascent technique with a fixed step size to determine the optimal value of the RIS phase shifts.

 Chu et al.~\cite{DBLP:IRS_MIMO_SR1} addressed the secrecy rate maximization problem in RIS-assisted systems by jointly designing the precoding matrix and RIS phase shifts using the block coordinate descent (BCD) and maximization–minimization (MM) algorithms. The authors in \cite{DBLP:IRS-MISO_secrecyrate_ncc} maximized the secrecy rate of RIS-aided multiple-input-single-output (MISO) systems by solving the real-valued optimization problem using a fixed step-size gradient ascent method. Li et al.~\cite{DBLP:IRS_MIMO_AIDED_Vechicular_secrecy_rate} considered a RIS-assisted MIMO vehicular network consisting of a multi-antenna Alice, a multi-antenna IU, and a multi-antenna Eve and attempted to maximize the secrecy rate of the system using semi-definite programming (SDP) and Riemannian manifold optimization (RMO) to design the active beamforming and phase shifts of the RIS elements. Shen et al.~\cite{DBLP:IRS-MISO_secrecy_rate_maximization} proposed an alternating optimization method to jointly design the transmit covariance matrix and RIS phase shift matrix to enhance the secrecy rate in a RIS-aided multi-antenna communications system.
  Several studies have investigated the secrecy outage probability (SOP) in RIS-assisted systems. For instance, Yang et al. \cite{DBLP:journals/tvt/YangYXHTR20} analyzed the SOP in RIS-aided SISO systems in which the NU locations were fixed, while Wang et al. \cite{DBLP:journals/icl/WangWWCY25} considered the case where the NUs were randomly located. Xiu et al.~\cite{DBLP:journals/icl/XiuZSZ21} investigated the secrecy rate maximization problem in RIS-assisted millimeter-wave (mmWave) communication systems under practical hardware constraints. Pala et al.~\cite{DBLP:journals/twc/PalaTKSLS24} extended the analysis to multiple RIS, multiple IUs, and multiple NUs in a multiple user (MU)-MIMO downlink scenario, addressing the secrecy rate maximization problem under discrete-phase constraints for the RIS elements.
 Finally, Wijewardena et al. ~\cite{DBLP:journals/icl/WijewardenaSHAE21} considered the secrecy rate maximization problem in RIS-aided SISO systems with two-way communication.

 \subsection{Motivations and Contributions}

 As indicated in Table~\ref{tab:related_work}, existing studies ~\cite{DBLP:journals/jsac/YuXSNS20, DBLP:journals/ojcs/AlmohamadTAFKHA20,DBLP:journals/tifs/ZhangDSAN21,  DBLP:journals/corr/abs-2103-16696,DBLP:IRS_MIMO_SR1, DBLP:IRS-MISO_secrecyrate_ncc, DBLP:IRS_MIMO_AIDED_Vechicular_secrecy_rate, DBLP:IRS-MISO_secrecy_rate_maximization,DBLP:journals/wcl/SirojuddinPH22,DBLP:journals/wcl/NingCCF20,DBLP:journals/icl/RanjanBMM23, DBLP:journals/tvt/YangYXHTR20,DBLP:journals/icl/WangWWCY25,DBLP:journals/icl/XiuZSZ21,DBLP:journals/twc/PalaTKSLS24,DBLP:journals/icl/WijewardenaSHAE21} primarily focus on designing the optimal phase shifts in RIS-assisted MIMO systems to maximize the secrecy rate. These studies operate under the assumption of an ideal (lossless) model, in which the magnitude of the RIS reflection coefficients is considered to be constant, irrespective of the phase shifts induced by the REs. However, this assumption does not hold in real-world scenarios, since the REs possess electrical resistance, and thus the magnitudes of the RIS reflection coefficients are not constant but vary with the phase shift~\cite{DBLP:journals/tcom/CaoYTP24}. In other words, although the ideal RIS model greatly simplifies the analysis, it fails to capture the true practical behavior of the RIS hardware.
 Motivated by this fact, the present study addresses the secrecy rate maximization problem in practical model RIS-assisted MIMO systems where the magnitude of the phase shift induced by the RE resistance must be explicitly considered.  
 To the best of the authors' knowledge, research on practical RIS models aimed at enhancing the PLS in RIS-assisted MIMO systems remains in an early stage. Motivated by this limitation, the present study examines how the adoption of a practical RIS model affects the secrecy rate compared with the ideal RIS model and proposes a low-complexity algorithm to solve the corresponding secrecy rate maximization problem. The contributions of this study are as follows:
\begin{itemize}
    \item Existing works consider an ideal RIS model with unit-modulus reflection coefficients. In contrast, the present study considers a practical RIS model in which the magnitude of each reflection coefficient depends on the phase shift of the corresponding RE. The model explicitly incorporates the inherent electrical resistance of the REs, resulting in a more realistic and accurate characterization of the RIS behavior in real-world systems. 
    \item The secrecy rate maximization problem is formulated for a practical RIS model, and the secrecy rate of the RIS-assisted MIMO system is obtained through the joint optimization of the transmit precoding matrix ($\TT$) and RIS phase shift ($\thetaa$). 
   \item In the objective function, $\TT$ and $\thetaa$ are coupled, and the magnitude of the RIS reflection coefficients has a nonlinear dependency on the phase shift. Accordingly, a low-complexity projected gradient method (PGM) is proposed in which $\TT$ and $\thetaa$ are simultaneously updated in each iteration.
    Note that although the PGM algorithm is well-known, its application to a specific problem remains challenging since the gradient calculation and projection are typically specific to the problem under investigation.
    \item Solving the exact secrecy rate maximization problem is computationally intensive. To address this issue, an alternative formulation is introduced that maximizes the channel power difference between the IU (Bob) and the NU (Eve).The proposed algorithm efficiently solves this alternative formulation with a substantially reduced computational complexity.
    \item The proposed optimization algorithms are general in the sense that they can be applied to both ideal (lossless) and practical (resistive) RIS models since they take into account the effect of the resistance parameter $R$ on the RIS characteristics, and the value of $R$ can be set to either zero (the ideal case) or a finite value (the practical case) as required. Therefore, the proposed algorithms can be applied to a wide range of system designs.
    \item  The simulation results demonstrate that the proposed PGM algorithms, when applied utilizing a practical RIS model, outperform the ideal RIS model, particularly in scenarios with randomly distributed RIS phase shifts. Moreover, due to their adaptive step-size initialization and control mechanism, the proposed algorithms achieve a higher secrecy rate and substantially lower computational cost than conventional fixed step-size PGM schemes.
     
    \end{itemize}
    
\subsection{Notations and Organization of Paper}
 The notations used throughout this paper are presented in Table~\ref{tab:notations}. Section II describes the system model and problem formulation for maximizing the secrecy rate in practical RIS-assisted MIMO systems. Section III describes the proposed low-complexity PGM algorithm to jointly solve the optimal values of the precoding matrix and RIS phase shifts. Section IV presents and discusses the simulation results. Finally, Section V summarizes the main findings of this study and presents some brief concluding remarks.

\begin{table*}[t!]
    \centering
    \renewcommand{\arraystretch}{1.25}
    \caption{Review of Relevant Works on RIS }
    \label{tab:related_work}

    \begin{tabular}{|p{1.4cm}|p{2cm}|p{2.6cm}|p{9.8cm}|}
    \hline
         \textbf{Reference}& \textbf{Ideal Model RIS} & \textbf{Practical Model RIS }   & \textbf{Performance metric/Study}\\ \hline
  \cite {DBLP:journals/wcl/SirojuddinPH22,DBLP:journals/wcl/NingCCF20}  & \checkmark  & \xmark   & Achievable rate maximization of RIS-aided MIMO systems\\\hline
        \cite{DBLP:journals/icl/RanjanBMM23} & \checkmark & \xmark & Achievable rate maximization of RIS-aided OFDM systems \\ \hline      
\cite {DBLP:IRS_MIMO_SR1} &  \checkmark & \xmark &  Maximization of secrecy rate of RIS-aided MIMO systems with artificial noise\\ \hline
        \cite {DBLP:IRS-MISO_secrecyrate_ncc} &  \checkmark  & \xmark & Maximization of secrecy rate of RIS-aided MISO systems\\ \hline
 \cite {DBLP:IRS_MIMO_AIDED_Vechicular_secrecy_rate} &\checkmark & \xmark &   Enhanced physical layer security in vehicular MIMO systems using RIS \\\hline
         \cite { DBLP:IRS-MISO_secrecy_rate_maximization}  & \checkmark  & \xmark  & RIS-enabled secure transmission in multi-antenna wireless networks \\\hline
       
        \cite{DBLP:journals/tvt/YangYXHTR20} & \checkmark &\xmark & Secrecy outage probability (SOP) analysis of RIS-aided MIMO systems for fixed location non-intended users\\ \hline
        \cite{DBLP:journals/icl/WangWWCY25} & \checkmark & \xmark &  SOP analysis of RIS-aided MIMO systems for randomly located non-intended users\\ \hline
        \cite{DBLP:journals/icl/XiuZSZ21} & \checkmark & \xmark &   Maximization of secrecy rate of RIS-assisted millimeter-wave (mmWave) communication system equipped with low-resolution digital-to-analog converters   \\ \hline
        \cite{DBLP:journals/twc/PalaTKSLS24} & \checkmark & \xmark &   Maximization of sum secrecy rate of secure RIS-assisted hybrid beamforming system with low resolution phase shifters  \\\hline
        \cite{DBLP:journals/icl/WijewardenaSHAE21} & \checkmark & \xmark &  Secrecy rate maximization of RIS-assisted two-way communication systems \\ \hline
          \textbf{Our work}  & \checkmark  & \checkmark & \textbf{Maximization of secrecy rate of practical model RIS-assisted MIMO systems}\\\hline
          
    \end{tabular}
\end{table*}

\begin{table*}[t!]
\renewcommand{\arraystretch}{1.25}
    \centering
    \caption{List of Notations}
    \label{tab:notations}
    \begin{tabular}{|l|l|l|l|l|l|l|}
        \hline
         \textbf{Notations} & \textbf{Representation} & \textbf{Notations} & \textbf{Representation} & \textbf{Notations}  & \textbf{Representation}\\\hline
          
$\mathfrak{R}$\{$\cdot$\} &  Real part  & $\in $ & Belongs to &  $\mathcal{\mathbb{C}}$ & Complex numbers\\\hline 
$\mathbb{E}\{\cdot\}$ & Expectation operator & $ \forall$ & Stands for all  & $\diag \left( \cdot \right)$ & Diagonalization operator \\\hline
$(\cdot)^{-1}$ & Matrix inverse  & $\mathcal{CN}(\cdot\; ,\cdot)$ &Complex Gaussian distribution & $(\cdot)^H$ & Hermitian  \\\hline
 $(\cdot)^\top$ & Transpose   &   $\left\| \cdot \right\|$ & Frobenius norm & $\left\lvert\cdot\right\rvert$ & Determinant \\\hline 
$\Delta (\cdot)$ & Del operator & $\triangleq$ &  Defined as & $\approx$ & Approximation  \\\hline
$(\cdot)^*$ & Conjugate operator &  $\circ $ & Hadamart product & $\mathbb{R}$ & Real numbers\\\hline   
\end{tabular}
\end{table*}

\section{System Model and Problem Formulation}
\label{System MOdel}
\begin{figure}[t!]
		\centering
\includegraphics[width=0.8\linewidth]{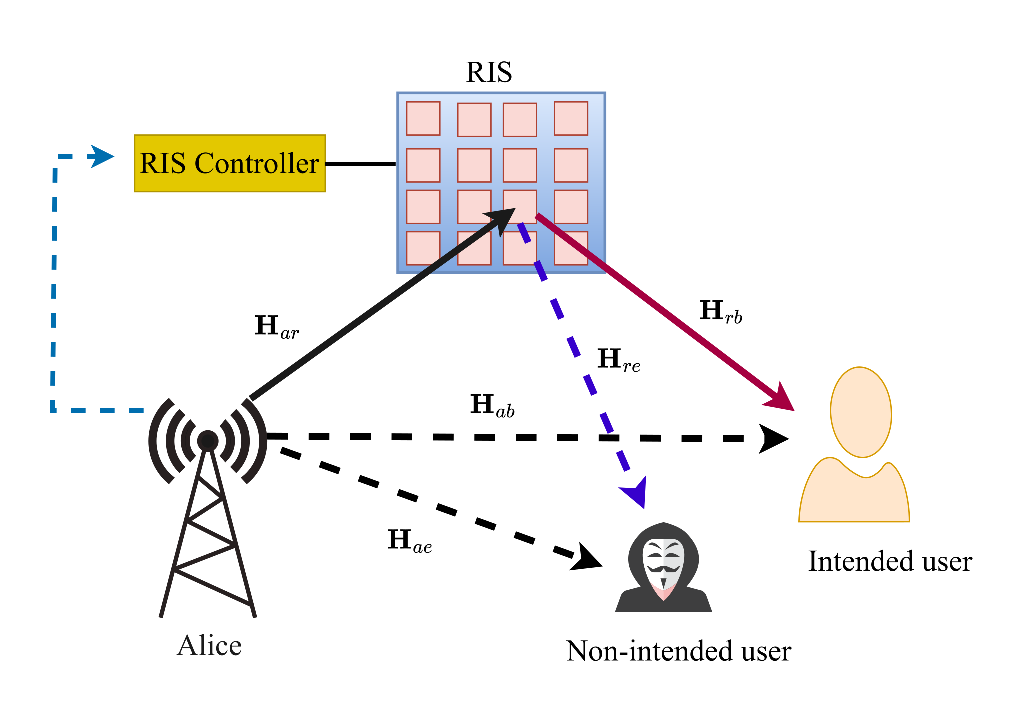}
		 \caption{{ System model of RIS-assisted MIMO systems}.}
			\label{fig:systemmodel}
	\end{figure}

\subsection{RIS-Assisted MIMO Systems}
This study considers a practical RIS-assisted MIMO wireless communication system in which a multi-antenna transmitter (Alice) communicates with an intended user (Bob) in the presence of an eavesdropper (Eve). Alice, Bob, and Eve are equipped with $\Na$, $\Nb$, and $\Ne$ antennas, respectively, and the transmission from Alice to Bob is assisted by an RIS consisting of $M$ passive REs, as shown in Fig. \ref{fig:systemmodel}. The corresponding channel state information (CSI) is given as follows: the channel between Alice and Bob (IU) is denoted as $\Hab \in \mathbb{C}^{\Nb \times \Na}$, between Alice and RIS as $\Har \in \mathbb{C}^{M \times \Na}$, between RIS and IU as $\Hrb \in \mathbb{C}^{\Nb \times M}$, between Alice and Eve (NU) as $\Hae \in \mathbb{C}^{\Ne \times \Na}$, and between RIS and NU as $\Hre \in \mathbb{C}^{\Ne \times M}$. Alice transmits a parallel data stream $\ss$ of size $\Ns$ with $\mathbb{E} \left\{ \ss \ss^H \right\} = \II_{\Ns}$. Before the data are transmitted to IU, a linear precoding matrix  $\TT \in \mathbb{C}^{\Na \times \Ns}$ is applied to exploit the available spatial multiplexing gain of the MIMO system.
For the system model shown in Fig.~\ref{fig:systemmodel}, the signals received at IU and NU are denoted as $\yb \in \mathbb{C}^{\Nb}$ and  $\ye \in \mathbb{C}^{\Ne}$, respectively, and are written as
\begin{equation} \label{eq:yb}
    \yb = \left( \Hab + \Hrb \Phii \Har \right) \TT \ss + \nb,
\end{equation}
\begin{equation} \label{eq:ye}
    \ye = \left( \Hae + \Hre \Phii \Har \right) \TT \ss + \ne,
\end{equation}
where $\Phii = \diag \left( \phi_1, \dots, \phi_M \right) \in \mathbb{C}^{M \times M}$; $\phi_m \in \mathbb{C}$ is the reflection coefficient of the $m$th RE; and $\nb \in \mathbb{C}^{\Nb}$ and $\ne \in \mathbb{C}^{\Ne}$ are the zero-mean complex Gaussian noise components at IU and NU with variances $\sigmab$ and $\sigmae$, respectively.

For simplicity, let $\Hhatb \triangleq \Hab + \Hrb \Phii \Har$ and $\Hhate \triangleq \Hae + \Hre \Phii \Har$ be the effective channels for the Alice-IU and Alice-{NU} links, respectively. In addition, let the achievable rates of the IU and NU be denoted as $\Rb$ and $\Re$, respectively, and be written as 
\begin{equation} \label{eq:Cb}
    \Rb = \log_2 \det \left( \II_{\Nb} + \Hhatb \TT \TT^H  \Hhatb^H / \sigmab \right).
\end{equation}
\begin{equation} \label{eq:Ce}
    \Re = \log_2 \det \left( \II_{\Ne} + \Hhate \TT \TT^H  \Hhate^H / \sigmae \right).
\end{equation}

\subsection{Practical RIS Phase Shift Model}

The reflection coefficient $\phi_m$ of the $m$th RE of the RIS  can be expressed as $\phi_m = \beta_m e^{j \theta_m}$, where $\beta_m \in \mathbb{R}$ and $\theta_m \in \mathbb{R}$ are the reflection magnitude and phase shift of the RE, respectively. This study considers a practical RIS phase shift model. Consequently, the reflection magnitude, $\beta_m$, is a function of the phase shift. 
As discussed in~\cite{DBLP:journals/tcom/AbeywickramaZWY20}, the reflection coefficient $\phi_m$ of the $m$th RE can be written as 
\begin{equation} \label {eq:practical_model1}
        \phi_m= \frac{Z_m(C_m,R_m)-Z_0}{Z_m(C_m,R_m)+Z_0}, \quad \forall m \in \mathcal{M},
\end{equation}
where $Z_0$ is the free-space impedance, and ${Z_m(C_m,R_m)}$ is the $m$th RE impedance, which can be expressed as
\begin{equation}\label {eq:practical_model2}
        Z_m(C_m,R_m)= \frac{j2\pi fL_1(j2\pi L_2+\frac{1}{j 2\pi f C_m}+R_m)}{j2\pi fL_1+(j2\pi L_2+\frac{1}{j 2\pi f C_m}+R_m)}, \, \forall m \in \mathcal{M},
\end{equation}
where $f$, $L_1$, $L_2$, $C_m$, and $R_m$, represent the signal frequency, inductance of the bottom layer, inductance of the top layer, tunable equivalent capacitance, and loss-related resistance in the equivalent circuit, respectively~\cite{DBLP:journals/icl/CaiLLL20}.

Following~\cite{DBLP:journals/tcom/AbeywickramaZWY20}, (\ref{eq:practical_model1}) and (\ref{eq:practical_model2}) can be restructured to obtain the following relationship between $\beta_m$ and  $\theta_m$: 
\begin{subequations} \label{eq:practical model3}
\begin{align} 
    &\beta_m \left( \theta_m \right)= \left( 1 - \betamin \right) \left( \frac { \sin \left( \theta_m - \tilde{\theta} \right) + 1}{2} \right)^{\alpha} + \betamin,\, \forall m  \\
    &\thetamin \leq \theta_m \leq \thetamax, \quad \forall m \in \mathcal{M}, \label{eq:theta boundary}
\end{align}
\end{subequations}
where $\betamin$, $\tilde{\theta}$, and $\alpha$ are constants representing the RE characteristic, with values determined by the specific electrical circuit implementation within the element. In contrast to \cite{DBLP:journals/tcom/AbeywickramaZWY20} which assumes a full $2 \pi$ phase shift range capability, the values of $\theta_m$ for the practical REs considered in the present study are restricted by $\thetamin \leq \theta_m \leq \thetamax$ with $\thetamin > -\pi$ and $\thetamax < \pi$. The relationship between $\beta_m$ and $\theta_m$, as modeled in (\ref{eq:practical model3}), is depicted in Fig. \ref{fig:practical model}.

\begin{figure}
    \centering
    \includegraphics[width=0.5\linewidth]{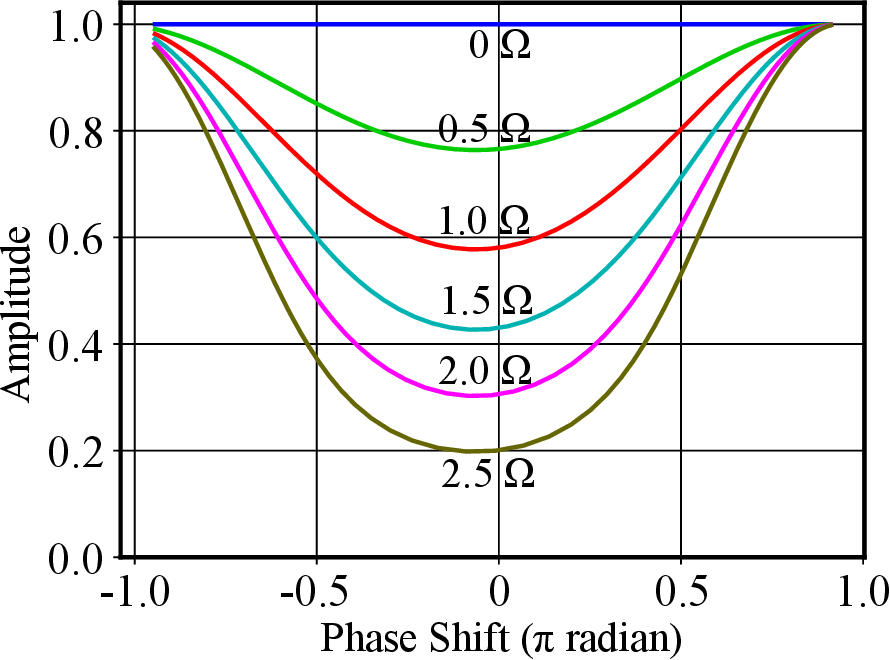}
    \caption{Relationship between $\beta_m$ and $\theta_m$ in practical RIS system as modeled by (\ref{eq:practical model3}) for different values of $R$.}
    \label{fig:practical model}
\end{figure}

\subsection{Problem Formulation}

The objective of this paper is to maximize the secrecy rate, i.e., $\Csec \triangleq (\Rb - \Re)^+$\footnote{Here, the notation $\{.\}^+=\max (0,.)$ denotes that the secrecy rate should always be positive. When the problem~(\ref{prob:main:obj}) is solved optimally, then the $\max(0,.)$ operator from the secrecy rate equ.~(\ref{prob:main:obj}) can be omitted without affecting the optimality~\cite{DBLP:journals/tcom/TaghizadehNMF19, DBLP:journals/twc/PalaTKSLS24}. } for the practical RIS-aided MIMO system shown in Fig.~\ref{fig:systemmodel}. 
The secrecy rate maximization problem can be formulated as follows:
\begin{subequations} \label{prob:main}
\begin{align} 
     \max_{\TT, \thetaa} \quad & \log_2 \det \left( \II_{\Nb} + \Hhatb \TT \TT^H  \Hhatb^H / \sigmab \right) \nonumber \\
    & - \log_2 \det \left( \II_{\Ne} + \Hhate \TT \TT^H  \Hhate^H / \sigmae \right), \label{prob:main:obj} \\
    \text{s.t.} \quad & \trace \left( \TT \TT^H \right) \leq P, \label{prob:main:const power} \\
    & \phi_m = \beta_m \, e^{j \theta_m}, \forall m \in \mathcal{M}, \label{prob:main:phim}\\
    & \beta_m = \left( 1 - \betamin \right) \left( \frac { \sin \left( \theta_m - \tilde{\theta} \right) + 1}{2} \right)^{\alpha} + \betamin, \forall m,\label{prob:main:betam}\\
    & \thetamin \leq \theta_m \leq \thetamax, \forall m \in \mathcal{M}. \label{prob:main:bound theta}
\end{align}
\end{subequations}
 In problem (\ref{prob:main}), constraint (\ref{prob:main:const power}) ensures that the total transmit power at Alice does not exceed the maximum allowed power $P$. Furthermore, (\ref{prob:main:phim}) states that, in the considered practical model, the reflection coefficients of the RE elements are functions of both the reflection magnitude and the phase shift. (\ref{prob:main:betam}) shows that the reflection magnitude depends nonlinearly on the phase shift of the corresponding RE. Finally, (\ref{prob:main:bound theta}) limits the range of the allowable phase shift $\thetaa_m$ for the $m$th RE. The main challenges encountered in solving the formulated optimization problem in (\ref{prob:main}) are described as follows:

\begin{enumerate}
  \item  Due to the imposed constraints and the non-convex nature of the objective function, the formulated optimization problem~(\ref{prob:main}) is challenging to solve. In addition, the precoding matrix $\TT$ and phase shift $\thetaa$ are coupled in the formulated problem, which adds to the complexity of the optimization process (\ref{prob:main}).
  
  \item The practical RIS model introduces a non-linear mapping between the phase shift $\theta_m$ and the reflection magnitude $\beta_m$ for the $m$th RE of the RIS, as shown in Fig.~\ref{fig:practical model}. Thus, the problem (\ref{prob:main}) is more complex than the ideal ( lossless) RIS model~\footnote{In ideal ( lossless) RIS model reflection magnitude i,e., $\beta_m=1 \forall m \in \mathcal{M}$.   }.
  
  \item Currently, no standard mathematical optimization technique is available to efficiently obtain a globally optimal solution for a non-convex optimization problem such as that in (\ref{prob:main}). Accordingly, the present study proposes a low-complexity optimization algorithm designed to obtain an efficient suboptimal solution by jointly designing the precoding matrix $\TT$ and RIS phase shift $\thetaa$. 
\end{enumerate}
 
The following section addresses the challenges outlined in $(1)$--$(3)$ and introduces the proposed low-complexity optimization algorithm. The proposed algorithm aims to achieve an effective trade-off between the computational cost and the secrecy-rate performance, demonstrating substantial performance improvements over existing state-of-the-art methods.

\section{Joint Precoding Matrix and Phase Shift Optimization} \label{sec:join opt}
In problem (\ref{prob:main}), although $\TT$ and $\thetaa$ are coupled in the objective function, they are separated in the constraints. This study proposes a joint optimization of the two parameters by simultaneously updating $\TT$ and $\thetaa$ in each iteration using PGM \cite{Perovic21Achievable}. Suppose that the values of $\TT$ and $\thetaa$ in the previous iteration of the optimization process are denoted as $\TT^{(i_1-1)}$ and $\thetaa^{(i_1-1)}$, respectively. By using PGM, in each iteration,
\begin{enumerate}
    \item $\TT^{(i_1-1)}$ and $\thetaa^{(i_1-1)}$ are shifted toward their gradient directions, yielding $\That$ and $\thetaahat$.
    \item $\That$ and $\thetaahat$ are then projected in accordance with their respective constraints, giving $\TT^{(i_1)}$ and $\thetaa^{(i_1)}$, respectively.
\end{enumerate}

In particular, in iteration $i_1$, $\That$ is calculated as
\begin{equation} \label{eq:T hat}
    \That = \TT^{(i-1)} + \upsilon_t \frac{\partial \Csec}{\partial \TT^*},
\end{equation}
where $\upsilon_t$ is the step-size, and the gradient of $\Csec$ with respect to $\TT$ is expressed as
\begin{align} \label{eq:dCsec dT}
    \frac{\partial \Csec}{\partial \TT^*} =& \frac{1}{\ln 2} \left[ \frac{1}{\sigmab} \Hhatb^H \AAb^{-1} \Hhatb \TT - \frac{1}{\sigmae} \Hhate^H \AAe^{-1} \Hhate \TT \right],
\end{align}
where $\AAb \triangleq \II_{\Nb} + \Hhatb \TT \TT^H \Hhatb^H / \sigmab$ and $\AAe \triangleq \II_{\Ne} + \Hhate \TT \TT^H \Hhate^H / \sigmae$. Note that $\That$ in (\ref{eq:T hat}) may violate the power constraint (\ref{prob:main:const power}), i.e., $\trace \left( \That \That^H \right) > P$. Under this condition, $\That$ needs to be projected to a set that fulfills the power constraint (\ref{prob:main:const power}). It is easily shown that this projection can be performed by the following equation:
\begin{equation} \label{eq:T project}
    \TT^{(i_1)} = \sqrt{\frac{P}{\trace \left( \That \That^H\right)}} \That.
\end{equation}
The procedure for updating the precoding matrix $\TT$ is illustrated in Fig. \ref{fig:precoderprojection}.
\begin{figure}
    \centering
    \includegraphics[width=0.45\linewidth]{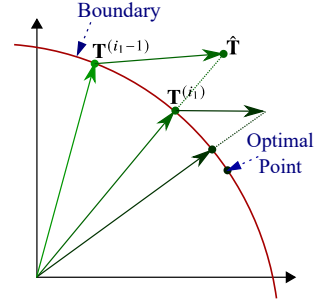}
    \caption{Illustration of precoder $\TT$ projection. Given $\TT^{(i_1-1)}$, $\That$ is obtained by shifting $\TT^{(1_1-1)}$ toward the direction of $\partial \Csec / \partial \TT^*$. Then, $\TT^{(i_1)}$ is obtained by projecting $\That$ to a set that satisfies (\ref{prob:main:const power}).}
    \label{fig:precoderprojection}
\end{figure}

The secrecy rate, $\Csec$, can be obtained as a closed-form function w.r.t. $\thetaa$ by substituting (\ref{prob:main:phim}) and (\ref{prob:main:betam}) into (\ref{prob:main:obj}). As in the $\TT$ update procedure, the solution of $\thetaahat$ in iteration $i_1$ can be calculated as
\begin{equation} \label{eq:thetaa hat}
    \thetaahat = \thetaa^{(i_1-1)} + \upsilon_{\theta} \frac{\partial \Csec}{\partial \thetaa},
\end{equation}
where $\upsilon_{\thetaa}$ is the step-size, and the gradient of $\Csec$ with respect to $\thetaa$ is given by
\begin{align} \label{eq:dCsec dthetaa}
    \frac{\partial \Csec}{\partial \thetaa} =& \frac{2}{\ln 2} \Ree \Bigg\{ \frac{\partial \phii}{\partial \thetaa} \nonumber \circ \diag \bigg[ \Har \TT \TT^H \Big( \Hhatb^H \AAb^{-1} \Hrb / \sigmab \\
    & - \Hhate^H \AAe^{-1} \Hre / \sigmae \Big) \bigg] \Bigg\},
\end{align}
where $\circ$ denotes the Hadamard product, and the operator $\diag[\XX]$ creates a vector whose elements are the diagonal of $\XX$. The detailed derivation for (\ref{eq:dCsec dT}) is presented in Appendix \ref{sec:derivations}. Meanwhile, $\partial \phii / \partial \thetaa \in \mathbb{C}^M$ is a vector in which the $m$th entry, \textit{i.e.,} $\partial \phi_m / \partial \theta_m$, is given in \cite{DBLP:IRS-MIMO_capacity_siro} as
\begin{align} \label{eq:dphim dthetam}
    \frac{\partial \phi_m}{\partial \theta_m} =& \frac{\left( 1 - \betamin \right) \alpha}{2^{\alpha}} \left[ \sin \left( \theta_m - \tilde{\theta} \right) + 1 \right]^{\alpha-1} \nonumber \\
    & \cdot \cos \left( \theta_m - \tilde{\theta} \right) e^{j \theta_m} + j \beta_m e^{j \theta_m}.
\end{align}

The entries of $\thetaahat$ in (\ref{eq:thetaa hat}) may violate the boundary constraint in (\ref{prob:main:bound theta}) since they can have any real value. In this situation, $\thetaahat$ must be projected onto a set that satisfies the boundary constraint (\ref{prob:main:bound theta}). This projection consists of several steps, as described below. First, $\hat\theta_m$ needs to be projected to the range $-\pi$ to $\pi$ by executing $\hat{\theta}_m \leftarrow  \left( (\hat{\theta}_m + \pi) \mod{2 \pi} \right) - \pi$. Secondly, if $\hat\theta_m$ falls within the feasible set, \textit{i.e.}, $\thetamin \leq \hat\theta_m \leq \thetamax$, then this value is set as the projected angle, \textit{i.e.}, $\theta_m^{(i_1)} = \hat\theta_m$. Otherwise, if $\hat\theta_m < \thetamin$ or $\hat\theta_m > \thetamax$, the projected angle is set to be either $\theta_m^{(i_1)} = \thetamin$ or $\theta_m^{(i_1)} = \thetamax$. The selection between these two boundary values is determined using a simple equation, as illustrated in Fig. \ref{fig:angleprojection}. For ease of explanation, let the following auxiliary variable $\breve{\theta}_m$ be introduced:
\begin{equation} \label{eq:theta breve}
    \breve{\theta}_m = \begin{cases}
        \hat{\theta}_m & \text{if } \hat{\theta}_m \geq 0 \\
        \hat{\theta}_m + 2 \pi & \text{if } \hat{\theta}_m < 0
    \end{cases}
    \quad \forall m.
\end{equation}
In other words, $\breve{\theta}_m$ is equivalent to $\hat{\theta}_m$, with a value ranging from $0$ to $2 \pi$. Furthermore, let us define $\thetac = \left( \left[ \thetamin + 2 \pi \right] + \thetamax  \right) / 2$ as the center point between $\thetamax$ and $\thetamin$, as illustrated in Fig. \ref{fig:angleprojection}. Hence, the projection is performed by executing the following equation:
\begin{equation} \label{eq:theta m project}
    \theta_m^{(i)} = \begin{cases} 
    \hat{\theta}_m & \text{if } \thetamin \leq \hat{\theta}_m \leq \thetamax, \\
    \thetamin & \text{if the first case is not satisfied and } \breve{\theta}_m \geq \thetac, \\
    \thetamax & \text{if the first case is not satisfied and } \breve{\theta}_m < \thetac.
\end{cases}
\end{equation}

\begin{figure}
    \centering
    \includegraphics[width=0.52\linewidth]{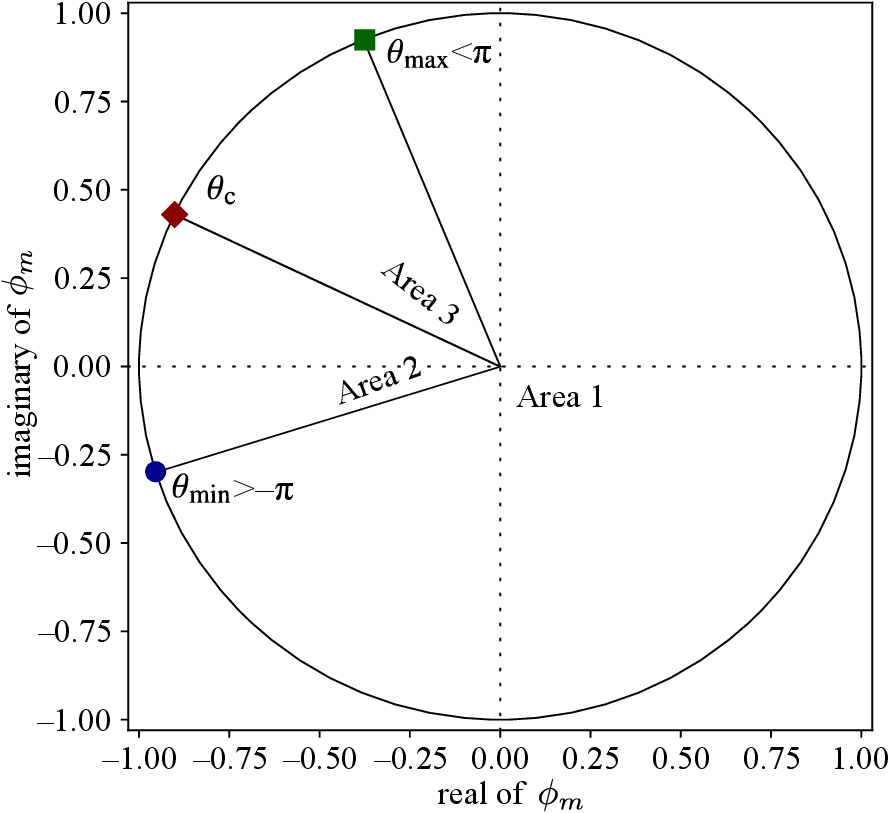}
    \caption{Illustration of the angle projection. $\theta_m^{(i)}$ is set to be $\theta_m^{(i)} = \hat{\theta}_m$, $\theta_m^{(i)} = \thetamin$ and $\theta_m^{(i)} = \thetamax$ when $\hat{\theta}_m$ is in Area 1, Area 2, and Area 3, respectively, as explained in (\ref{eq:theta m project}).}
    \label{fig:angleprojection}
\end{figure}

The proposed PGM approach, both for updating $\TT$ (executing (\ref{eq:T hat}) and then (\ref{eq:T project})) and updating $\thetaa$ (executing (\ref{eq:thetaa hat}) and then (\ref{eq:theta m project})), may not be stable since $\Csec$ may degrade over successive iterations, \textit{i.e.,} $\Csec^{(i)} < \Csec^{(i-1)}$ due to inappropriate choices for the step sizes $\upsilon_t$ and $\upsilon_{\theta}$ in (\ref{eq:T hat}) and (\ref{eq:thetaa hat}), respectively. To address this issue, the present work adopts a simple procedure to monitor both step sizes during the iteration process to ensure monotonicity. In particular, following the procedure in \cite{Judd22Sumrate}, if $\Csec^{(i)} < \Csec^{(i-1)}$, $\upsilon_t$ or $\upsilon_{\theta}$ (or both) are reduced iteratively by a factor $c$, where $0 < c < 1$, until $\Csec^{(i)} > \Csec^{(i-1)}$.

In addition to monitoring $\upsilon_t$ and $\upsilon_{\theta}$ during iterations to ensure monotonicity, it is also necessary to determine appropriate settings of their initial values. If the step sizes are too small, the iteration process may take to long to reach a stationary point. Conversely, if they are too large, iterations may drive $\TT$ and $\thetaa$ far away from their respective constraints and cause instability. An initial value for $\upsilon_{\theta}$ in the first iteration can be determined as
\begin{equation} \label{eq:init upsilon theta}
    \upsilon_{\theta} = \thetamaxstep / \max \left( \left| \frac{\partial \Csec}{\partial \theta_1} \right|, \cdots, \left| \frac{\partial \Csec}{\partial \theta_M} \right| \right).
\end{equation}
By executing (\ref{eq:init upsilon theta}), the maximum entry of vector $\left| \thetaahat - \thetaa^{\text{(init)}} \right|$ after executing (\ref{eq:thetaa hat}) when $i_1=1$ is $\thetamaxstep$. Exploiting the fact that $\hat{\theta}_m$ should remain between $\thetamin$ and $\thetamax$, $\thetamaxstep$ can be chosen from, for example, a value between 0 and $\left( \thetamax - \thetamin \right) / 2$. 

However, the proposed strategy for determining an initial value for $\upsilon_{\theta}$, as stated in (\ref{eq:init upsilon theta}), cannot be adopted for $\upsilon_t$ as $\TT$ does not have an elemental constraint. Accordingly, this study determines an initial value for $\upsilon_t$ by leveraging the fact that $\TT$ must lie on the boundary constraint $\trace \left( \TT \TT^H \right) = P$. As a result, the gradient shift in (\ref{eq:T hat}) must bring $\That$ 'not too far' from $\TT^{(0)}$. Accordingly, given $\TT^{(0)}$, which lies on the boundary constraint, the initial value for $\upsilon_{t}$ is determined as
\begin{equation} \label{eq:init upsilon t}
    \upsilon_t = \frac{\tau \sqrt{P} }{\left\| \partial C / \partial \TT^* \right\|},
\end{equation}
where $\tau > 0$, and $\left\| \cdot \right\|$ is the Frobenius norm operator. By executing (\ref{eq:init upsilon t}), the value of $\left\| \That - \TT^{(0)} \right\|$ after executing (\ref{eq:T hat}) in the first iteration is equal to approximately $\tau \left\| \TT^{(0))} \right\|$. The value of $\tau$ can be set, for example, between 0 and 1 to limit the amount of the gradient step in (\ref{eq:T hat}) in the first iteration. Algorithm ~\ref{alg:PGM} shows the details of the proposed PGM method for solving (\ref{prob:main}).

\begin{algorithm}
    \caption{PGM procedure for solving problem (\ref{prob:main})}
    \label{alg:PGM}
    \SetAlgoLined
    \KwIn{$\thetaa^{\text{(init)}}$, $\TT^{\text{(init)}}$, $\Hab$, $\Har$, $\Hrb$, $\Hae$, $\Hre$, $P$, $\beta_{\text{min}}$, $\alpha$, $\tilde{\theta}$, $\thetamin$, $\thetamax$, $\sigmab$, $\sigmae$.}
    \KwOut{$\Csec$, $\thetaaopt$, $\Topt$.}
    Initialize $i_1=0$, $|\Delta \Csec| = \infty$. \\
    \While{$i_1 < i_{\text{\normalfont{max}}}$ \normalfont{ and} $|\Delta \Csec| > \xi$} {
        Compute $\partial \Csec / \partial \thetaa$ using (\ref{eq:dCsec dthetaa}). \label{alg:PGM:dC dthetaa} \\
        Initialize $i_2$, $|\Delta \Csec| = \infty$. \\
        \If{$i_1=1$}{
            Compute $\upsilon_{\theta}$ using (\ref{eq:init upsilon theta}).
        } \label{alg:PGM:endif upsilontheta}
        \While{$i_2 < i_{ \normalfont{\text{max}}}$ \normalfont{ and } $|\Delta \Csec| > \xi$}{
            Compute $\thetaahat$ using (\ref{eq:thetaa hat}). \label{alg:PGM:thetaahat} \Comment{gradient-based update} \\
            Compute $\thetaa^{(i_1)}$ using (\ref{eq:theta m project}). \Comment{projection} \\
            Compute $\Delta \Csec = \Csec \left( \thetaa^{(i_1)} \right) - \Csec \left( \thetaa^{(i_1-1)} \right)$. \\
            \eIf {$\Delta C < 0$}{
                $\upsilon_{\thetaa} \leftarrow c \cdot \upsilon_{\thetaa}$. \Comment{step-size reduction}
            }{
            Break. \Comment{out from inner loop}
            } \label{alg:PGM:endif reduct theta}
        }
        Compute $\partial \Csec / \partial \TT^*$ using (\ref{eq:dCsec dT}). \label{alg:PGM:dC dT}\\
        Initialize $i_2$, $|\Delta \Csec| = \infty$. \\
        \If{$i_1=1$}{
            Compute $\upsilon_{t}$ using (\ref{eq:init upsilon t}).
        } \label{alg:PGM:endif upsilont}
        \While{$i_2 < i_{ \normalfont{\text{max}}}$ \normalfont{ and } $|\Delta \Csec| > \xi$}{
            Compute $\That$ using (\ref{eq:T hat}). \label{alg:PGM:That} \Comment{gradient-based update} \\
            Compute $\TT^{(i_1)}$ using (\ref{eq:T project}). \Comment{projection} \\
            Compute $\Delta \Csec = \Csec \left( \TT^{(i_1)} \right) - \Csec \left( \TT^{(i_1-1)} \right)$. \\
            \eIf {$\Delta C_{\normalfont{\text{sec}}} < 0$}{
                $\upsilon_t \leftarrow c \cdot \upsilon_t$. \Comment{step-size reduction}
            }{
            Break. \Comment{out from inner loop}
            } \label{alg:PGM:endif reduct t}
        }
    }
\end{algorithm}

To assess the implementation cost, the computational complexity of Algorithm \ref{alg:PGM} is derived in the following. Each complex multiplication is equivalent to four real multiplications; hence, any expression denoting the number of complex multiplications is scaled by a factor of four to obtain the real multiplication count. A Real$\times$complex scaling operation costs cost 2 real multiplications. Hermitian matrix updates exploit symmetry (reducing the cost by $\approx$0.5) but still retain the fourfold factor per complex multiplication. In general, the multiplication order among multiple matrices affects the resulting complexity, as in the matrix-chain multiplication problem. Since the value of $M$ is much larger than $\Na$, $\Nb$, and $\Ne$ in the practical setting considered in the present work, the multiplication order is determined such that factor $M$ appears as rarely as possible in the resulting complexity calculation. The other intermediate matrix multiplications are executed in order since no general assumption is made regarding the relative magnitudes of $\Na$, $\Nb$, and $\Ne$.

\begin{table}[t!]
\renewcommand{\arraystretch}{1.25}
    \centering
    \caption{Number of multiplications for various variables in Alg. \ref{alg:PGM}}
    \label{tab:complexity}
    \begin{tabular}{|p{0.16 \linewidth}|p{0.16 \linewidth}|p{0.53 \linewidth}|}
        \hline
        \textbf{Variable} & \textbf{Given} & \textbf{Number of multiplications} \\\hline         
        $\Hhatb$ & - & $4 (M \Nb + M \Na \Nb)$  \\\hline 
        $\Hhate$ & - & $4 (M \Ne + M \Na \Ne)$ \\\hline
        $\AAb^{-1}$ & $\Hhatb$ & $4 \big[ \Na \Nb \Ns + 0.5 \Nb^2 \Ns + 0.5 \Nb^2 + (1/3) \Nb^3 \big] + \lambdadiv$  \\\hline
        $\AAe^{-1}$ & $\Hhate$ & $4 \big[ \Na \Ne \Ns + 0.5  \Ne^2 \Ns + 0.5 \Ne^2 + (1/3) \Ne^3 \big] + \lambdadiv$ \\\hline
        $\betaa$ & - & $M \left( \lambdasin + \lambdapow(\alpha) + 2 \right)$ \\\hline
        $\partial \phii / \partial \thetaa$ & $\betaa$ & $M \big[ (2\lambdasin + 2 \lambdacos) + \lambdapow (\alpha -1) + 6 \big] + \lambdapow (\alpha) + 2$ \\\hline
        $\partial \Csec / \partial \thetaa$ & $\Hhatb$, $\Hhate$, $\AAb^{-1}$, $\AAe^{-1}$, $\partial \phii / \partial \thetaa$ & $4 \big[ 0.5 \Na^2 \Ns + M \Na^2 + \Na \Nb^2 + M \Na \Nb + \Na \Ne^2 + M \Na \Ne + M \Na + 1.25M \big] + \lambdadiv + 1$ \\\hline
        $\upsilon_{\theta}$ in (\ref{eq:init upsilon theta}) & $\partial \Csec / \partial \thetaa$ & $\lambdadiv$ \\\hline
        $\hat\thetaa$ in (\ref{eq:thetaa hat}) & $\partial \Csec / \partial \thetaa$ & $M$ \\\hline
        $\thetaa^{(i_1)}$ in (\ref{eq:theta m project}) & $\hat{\thetaa}$ & $M \big( \lambdadiv + 2\big)$ \\\hline
        $\Delta \Csec$ & $\Hhatb$, $\Hhate$ & $ 4 \big[ \Na \Nb \Ns + 0.5 \Nb^2 \Ns + 0.5 \Nb^2 + (1/3) \Nb^3 + \Na \Nb \Ne + 0.5 \Ne^2 \Ns + 0.5 \Ne^2 + (1/3) \Ne^3 \big] + \Nb \lambdalogtwo + \Ne \lambdalogtwo + 2$  \\\hline
        $\partial \Csec / \partial \TT^*$ & $\Hhatb$, $\Hhate$, $\AAb^{-1}$, $\AAe^{-1}$ & $4 \big[ \Na \Nb^2 + \Na^2 \Nb + \Na \Ne^2 + \Na^2 \Ne + 2 \Na^2 \Ns + 0.5 \Na \Ns \big] + \lambdadiv$ \\\hline
        $\upsilon_{t}$ in (\ref{eq:init upsilon t})& $\partial \Csec / \partial \TT^*$ & $(6 + 2 \lambdadiv) \Na \Ns + \lambdasqrt $ \\\hline
        $\That$ in (\ref{eq:T hat}) & $\partial \Csec / \partial \TT^*$ & $2 \Na \Ns$ \\\hline
        $\TT^{(i)}$ in (\ref{eq:T project}) & $\That$ & $6\Na \Ns + \lambdadiv + \lambdasqrt$ \\\hline
        
    \end{tabular}
\end{table}

Table \ref{tab:complexity} shows the number of multiplications required for the variables involved in Algorithm \ref{alg:PGM}, where the order of presentation of these variables is bottom-up (from the simplest to the more complex) and follows their appearance in the algorithm. In the table, factor 1/3 in the inverse operation assumes computation via Cholesky decomposition and triangular solves, rather than an explicit matrix inversion, which would introduce a factor of 4/3. $\lambdasin$, $\lambdacos$, $\lambdapow (\cdot)$, $\lambdadiv$, $\lambdalogtwo$, and $\lambdasqrt$ represent the equivalent number of scalar multiplications required to implement the sine, cosine, power, division, log$_2$, and square root operations, respectively, with the exact values of these variables depending on the specific machine implementation \cite{Johansson15Efficient}.

Based on Table \ref{tab:complexity}, the complexities of executing lines \ref{alg:PGM:dC dthetaa}--\ref{alg:PGM:endif upsilontheta}, \ref{alg:PGM:thetaahat}--\ref{alg:PGM:endif reduct theta}, \ref{alg:PGM:dC dT}--\ref{alg:PGM:endif upsilont}, and \ref{alg:PGM:That}--\ref{alg:PGM:endif reduct t} in Algorithm \ref{alg:PGM} are asymptotically $\mathcal{O}(\Gamma_{1 \theta})$, $\mathcal{O}(\Gamma_{2 \theta})$, $\mathcal{O}(\Gamma_{1 \text{t}})$, $\mathcal{O}(\Gamma_{2 \text{t}})$, respectively, where
\begin{subequations} \label{eq:Gammas}
\begin{align} 
    \Gamma_{1 \theta} =& M \left( \Na^2 + \Na \Nb + \Na \Ne \right) + \Na^2 \Ns + \Na \Nb^2 + \Na \Ne^2 \nonumber \\
    & + \Na \Nb \Ns + \Nb^2 \Ns + \Nb^3 + \Na \Ne \Ns + \Ne^2 \Ns + \Ne^3 \\
    \Gamma_{2 \theta} =& M \Na \left( \Nb + \Ne \right) + \Na \Nb \Ns + \Nb^2 \Ns + \Nb^3 \nonumber \\
    & + \Na \Nb \Ne + \Ne^2 \Ns + \Ne^3 \\
    \Gamma_{1 \text{t}} =& \Na \Nb^2 + \Na^2 \Nb + \Na \Ne^2 + \Na^2 \Ne + \Na^2 \Ns + \Na \Nb \Ns \nonumber \\
    & + \Nb^2 \Ns + \Nb^3 + \Na \Ne \Ns + \Ne^2 \Ns + \Ne^3 \label{eq:Gamma one theta}\\
    \Gamma_{2 \text{t}} =& \Na \Nb \Ns + \Nb^2 \Ns + \Nb^3 + \Na \Nb \Ne + \Ne^2 \Ns + \Ne^3 \label{eq:Gamma two theta}
\end{align}
\end{subequations}
In (\ref{eq:Gamma one theta}) and (\ref{eq:Gamma two theta}), it is assumed that $\Hhatb$ and $\Hhate$ are stored before the execution of line \ref{alg:PGM:dC dT} and are thereform unaffected by $M$. Based on the notations in (\ref{eq:Gammas}), the complexity of Algorithm \ref{alg:PGM} is asymptotically $\mathcal{O} \left( I_1 \left( \Gamma_{1 \theta} + \Gamma_{1\text{t}} + I_{2 \theta} \Gamma_{2 \theta} + I_{2 \text{t}} \Gamma_{2\text{t}} \right) \right)$, where $I_1$, $I_{2 \theta}$, and $I_{2 \text{t}}$ represent the average iteration numbers required to execute the outer loop, the inner loop for updating $\thetaa$, and the inner loop for updating $\TT$, respectively. In practical deployments, $M$ is usually much larger than $\Na$, $\Nb$, and $\Ns$. Let us assume that $\Na = \Nb = \Ne = \Nx$. The overall complexity is thus asymptotically $\mathcal{O} \left( I_1 \left( M \Nx^2 + \Nx^3 + I_{2 \theta} (M \Nx^2 + \Nx^3) + I_{2 \text{t}} \Nx^3 \right) \right)$. Consequently, Algorithm \ref{alg:PGM} exhibits low complexity, as it grows polynomially with the matrix dimensions (at most cubic in order) and linearly with respect to $M$. This linear dependence on $M$ is particularly advantageous, since $M$ is typically much larger than the other system parameters in practical deployments.

\section{Sub-optimal Phase Shift with CPDM} \label{sec:CPDM}
The secrecy rate maximization problem can be interpreted as maximizing $\Rb$ while simultaneously minimizing $\Re$. Intuitively, $\Rb$ is high if the channel quality from Alice to {IU}, $\Hhatb$, is high. Similarly, $\Re$ is low if the channel quality from Alice to {NU}, $\Hhate$, is low. Motivated by this observation, this study proposes a channel power difference maximization (CPDM) formulation as a sub-problem to find a sub-optimal solution $\thetaa^{\text{(sub)}}$ that provides substantially lower complexity than that of the main secrecy rate maximization problem (\ref{prob:main}). CPDM is formulated as follows:
\begin{subequations} \label{prob:CPDM}
\begin{align} 
    \max_{\TT, \thetaa} \quad & \Pdiff = \trace \left[ \TT^H \left( \frac{1}{\sigmab} \Hhatb^H \Hhatb - \frac{1}{\sigmae} \Hhate^H \Hhate \right) \TT  \right], \label{prob:CPDM:obj} \\
    \text{s.t.} \quad & \trace \left( \TT^H \TT \right) \leq P \\
    & \phi_m = \beta_m e^{j \theta_m}, \forall m, \label{prob:CPDM:phim}\\
    & \beta_m = \left( 1 - \betamin \right) \left( \frac { \sin \left( \theta_m - \tilde{\theta} \right) + 1}{2} \right)^{\alpha} + \betamin, \forall m, \label{prob:CPDM:betam}\\
    & \thetamin \leq \theta_m \leq \thetamax, \forall m \label{prob:CPDM:bound theta}
\end{align}
\end{subequations}

The CPDM criterion is analytically related to the main secrecy rate maximization criterion, as established in the following proposition.

\begin{proposition} \label{prop:CPDM}
    The CPDM criterion expressed in (\ref{prob:CPDM:obj}) is a surrogate function of the secrecy rate criterion in (\ref{prob:main:obj}) and serves as an upper bound when $\FFb - \FFe \succeq \mathbf{0}$, where $\FFb \triangleq \TT^H \Hhatb^H \Hhatb \TT / \sigmab$ and $\FFe \triangleq \TT^H \Hhate^H \Hhate \TT / \sigmae$. Moreover, the bound gap is negligible when the eigenvalues of $\FFb$ and $\FFe$ are small.
\end{proposition}

\begin{IEEEproof}
    The proof is presented in Appendix \ref{sec:CPDM proof}.
\end{IEEEproof}

It is noted that (\ref{prob:CPDM:obj}) contains a matrix trace operation with a lower complexity than that of the matrix determinant operation in (\ref{prob:main:obj}).

Given a fixed $\thetaa$, let $\Pdiff$ in (\ref{prob:CPDM:obj}) be expressed as $\Pdiff = \trace \left( \TT^H  \GG \TT\right)$, where $\GG \triangleq \Hhatb^H \Hhatb / \sigmab - \Hhate^H \Hhate / \sigmae$. If $\GG \succeq \mathbf{0}$, the optimal $\TT$ is a matrix whose columns are the $\Ns$ eigenvectors of $\GG$ corresponding to the $\Ns$ largest eigenvalues. The resulting matrix is then normalized to satisfy $\trace \left( \TT^H \TT \right) = P$. Otherwise, if $\GG$ is indefinite, the optimal $\TT$ comprises the eigenvectors corresponding to the $\Ns$ largest positive eigenvalues and avoids the negative modes (which correspond to directions where the NU channel channel is stronger than the IU channel after noise weighting).

Given a fixed $\TT$, this study proposes a PGM method to solve problem (\ref{prob:CPDM}), in which $\thetaa$ is repeatedly shifted with respect to the gradient direction, and is projected each time in such a way that it meets the constraint (\ref{prob:CPDM:bound theta}). $\Pdiff$ can be transformed to a closed-form function w.r.t. $\thetaa$ by substituting (\ref{prob:CPDM:phim}) and (\ref{prob:CPDM:betam}) into (\ref{prob:CPDM:obj}). In iteration $i_1$, $\thetaa$ is updated as
\begin{equation} \label{eq:thetaa bar}
    \bar{\thetaa} = \thetaa^{(i_1-1)} + \bar{\upsilon}_{\theta} \frac{\partial \Pdiff}{\partial \thetaa},
\end{equation}
where $\bar{\upsilon}_{\theta}$ is the step size, and the gradient of $\Pdiff$ with respect to $\thetaa$ is given by
\begin{align} \label{eq:dPdiff dthetaa}
    \frac{\partial \Pdiff}{\partial \thetaa} =& 2 \Ree \Bigg\{ \frac{\partial \phii}{\partial \thetaa} \circ \diag \bigg[ \Har \TT \TT^H \Big( \Hhatb^H \Hrb / \sigmab \nonumber \\
    & - \Hhate^H \Hre / \sigmae\Big) \bigg] \Bigg\}.
\end{align}
Next, $\thetaa^{(i_1)}$ is obtained using the projection procedure, as described in (\ref{eq:theta m project}). The step-size initialization and monitoring procedure described in Sec. \ref{sec:join opt} can also be applied to enhance the robustness of the algorithm, where the equation for determining the initial value of $\bar{\upsilon}_{\theta}$ is given by
\begin{equation} \label{eq:init upsilon bar theta}
    \bar{\upsilon}_{\theta} = \thetamaxstep / \max \left( \left| \frac{\partial \Pdiff}{\partial \theta_1} \right|, \cdots, \left| \frac{\partial \Pdiff}{\partial \theta_M} \right| \right).
\end{equation}
Algorithm \ref{alg:CPDM} presents the details of the PGM method proposed in this study to solve problem (\ref{prob:CPDM}) after performing the step-size initialization and monitoring procedure described in Sec. \ref{sec:join opt}.

\begin{algorithm}
    \caption{PGM procedure for solving problem (\ref{prob:CPDM}) Given $\TT$}
    \label{alg:CPDM}
    \SetAlgoLined
    \KwIn{$\thetaa^{\text{(init)}}$, $\TT$, $\Hab$, $\Har$, $\Hrb$, $\Hae$, $\Hre$, $\beta_{\text{min}}$, $\alpha$, $\tilde{\theta}$, $\thetamin$, $\thetamax$.}
    \KwOut{$\Pdiff$, $\thetaa^{\text{(sub)}}$.}
    Initialize $i_1=0$, $|\Delta \Pdiff| = \infty$. \\
    \While{$i_1 < i_{\text{\normalfont{max}}}$ \normalfont{ and} $|\Delta \Pdiff| > \xi$} {
        Compute $\partial \Pdiff / \partial \thetaa$ using (\ref{eq:dPdiff dthetaa}). \label{alg:CPDM:dP dthetaa} \\
        Initialize $i_2$, $|\Delta \Pdiff| = \infty$. \\
        \If{$i_1=1$}{
            Compute $\bar{\upsilon}_{\theta}$ using (\ref{eq:init upsilon bar theta}).
        } \label{alg:CPDM:endif upsilon init}
        \While{$i_2 < i_{ \normalfont{\text{max}}}$ \normalfont{ and } $|\Delta \Pdiff| > \xi$}{
            Compute $\bar{\thetaa}$ using (\ref{eq:thetaa bar}). \label{alg:CPDM:theta bar} \Comment{gradient-based update} \\
            Compute $\thetaa^{(i_1)}$ using (\ref{eq:theta m project}). \Comment{projection} \\
            Compute $\Delta \Pdiff = \Pdiff \left( \thetaa^{(i_1)} \right) - \Pdiff \left( \thetaa^{(i_1-1)} \right)$. \\
            \eIf {$\Delta P_{\normalfont{\text{diff}}} < 0$}{
                $\bar{\upsilon}_{\thetaa} \leftarrow c \cdot \bar{\upsilon}_{\thetaa}$. \Comment{step-size reduction}
            }{
            Break. \Comment{out from inner loop}
            } \label{alg:CPDM:endwhile}
        }
    }
\end{algorithm}
To assess the implementation cost, the computational complexity of Algorithm~\ref{alg:CPDM} is derived in the following. Utilizing the same procedure as that employed for computing the complexity of Algorithm~\ref{alg:PGM}, Table \ref{tab:complexity2} presents the number of multiplications required for the variables involved in Algorithm~\ref{alg:CPDM}. (Note that some of the variables are omitted since they have already been given in Algorithm~\ref{alg:PGM}).) Based on Table \ref{tab:complexity2}, the complexities of executing lines \ref{alg:CPDM:dP dthetaa}--\ref{alg:CPDM:endif upsilon init} and \ref{alg:CPDM:theta bar}--\ref{alg:CPDM:endwhile} are asymptotically $\mathcal{O} \left( \bar{\Gamma}_{1 \theta} \right)$ and $\mathcal{O} \left( \bar{\Gamma}_{2 \theta} \right)$, respectively, where
\begin{align}
    \bar{\Gamma}_{1 \theta} &= \Na^2 \Ns + M \Na^2 + M \Na \Nb + M \Na \Ne \\
    \bar{\Gamma}_{2 \theta} &= M \Na \Nb + M \Na \Ne.
\end{align}
Therefore, the complexity of Algorithm~\ref{alg:CPDM} is asymptotically $\mathcal{O} \left( \bar{I}_1 \left( \bar{\Gamma}_{1 \theta} + \bar{I}_2 \bar{\Gamma}_{2 \theta} \right) \right) \approx \mathcal{O} \left( \bar{I}_1 \left( M \Nx^2 + \Nx^3 + \bar{I}_2 M \Nx^2 \right) \right)$, where $\bar{I}_{1}$ and $\bar{I}_2$ are the average iteration numbers required to execute the outer and inner loops, respectively. As seen, the complexity is a third-degree polynomial with respect to the matrix size. In addition, it is only linear with $M$, where the value of $M$ is large compared to $\Nx$. 

\begin{table}[t!]
\renewcommand{\arraystretch}{1.25}
    \centering
    \caption{Number of multiplications for various variables in Alg.~\ref{alg:CPDM}}
    \label{tab:complexity2}
    \begin{tabular}{|p{0.16 \linewidth}|p{0.16 \linewidth}|p{0.53 \linewidth}|}
        \hline
        \textbf{Variable} & \textbf{Given} & \textbf{Number of multiplications} \\\hline         
        $\partial \Pdiff / \partial \thetaa$ & $\Hhatb$, $\Hhate$, $\partial \phii / \partial \thetaa$ & $4 (0.5 \Na^2 \Ns + M \Na^2 + M \Na \Nb + M \Na \Ne + 2 M \Na + 1.25M )$  \\\hline
        $\bar{\upsilon}_{\theta}$ in (\ref{eq:init upsilon bar theta}) & $\partial \Pdiff / \partial \thetaa$ & 1 \\\hline
        $\bar{\thetaa}$ & $\partial \Pdiff / \partial \thetaa$ & $M$ \\\hline
        $\Delta \Pdiff$ & $\Hhatb$, $\Hhate$ & $4 \left( \Na \Nb + \Na \Ne \right)$ \\\hline
    \end{tabular}
\end{table}

\section{Numerical Results and Discussion}
\label{simulation}

This section presents the numerical simulation results, which demonstrate the efficacy of the proposed PGM approach in terms of the secrecy rate and channel power difference for the practical RIS-assisted MIMO system shown in Fig.~\ref{fig:systemmodel}. Referring to Fig.~\ref{fig:postion}, the positions of Alice, the RIS, IU, and NU are assumed to be (0,5,10), (100,0,2), (100,3,0), and (90,2,0), respectively, as described in \cite{DBLP:IRS_MIMO_SR1}. Note that all the distances are measured in meters (m). In implementing the system model shown in Fig.~\ref{fig:systemmodel}, Alice is assumed to possess $ N_a=4$ antennas, while IU and NU have the same number of antennas ($N_b=N_e=4$).  The large-scale path loss can be evaluated as $\rho_l =\rho_{l_o} (\frac{d}{d_o})^{-\gamma}  $, where $\rho$ is the path loss at distance $d$, $\rho_{l_o}=-30$ dB is the path loss at distance $d_o=1$ m, and $\gamma$ is the path loss exponent, with its values for $\Har, \Hrb, \Hre$ $\Hab$, and $\Hae$ selected as $2.2,2.5,2.5,3.5$, and $3.5$, respectively~\cite{DBLP:simulation_alpha_set}. The transmission power available at Alice is $P=30$ dBm, and IU and NU both have noise powers of $\sigmab=\sigmae=\sigma^2= -110$ dBm. Finally, for the practical RIS model in (\ref{eq:practical_model1}) and (\ref{eq:practical_model2}), the parameter values $L_1=2.5$ nh, $L_2=0.7$ nh, $f=2.5$ GHz are set directly, while $\beta_{\text{min}},\tilde{\theta}, \alpha,\theta_{\text{min}}$, and $\theta_{\text{max}}$ are obtained using the method described in~\cite{DBLP:journals/tcom/AbeywickramaZWY20}. 

The present study considers a practical RIS system. In other words, the proposed algorithms explicitly incorporate the resistance values ($R$) of the REs during the optimization process. Notably, however, Algorithms \ref{alg:PGM} and \textbf{\ref{alg:CPDM}} are both general in nature, and can thus be equally applied to ideal RIS scenarios, i.e., $R=0$$\Omega$.

\subsection{Benchmark Schemes}

To the best of the authors' knowledge, no prior work has addressed all the components outlined in Section II. Therefore, to compare the efficacy of the proposed PGM algorithms, the following benchmark techniques are employed:
\begin{enumerate}

   \item  Practical PGM: In this benchmark, proposed Algorithm ~\ref{alg:PGM} is used to jointly design the optimal values of $\thetaa$ and $\bm{T}$ and calculate $\Csec$ with real RE resistance values $(R)$.
   
   \item Ideal PGM: In this scheme, Algorithm ~\ref{alg:PGM} is first applied under the ideal RIS assumption i.e., $R=0\Omega$ to obtain the optimal $\thetaa$ and $\TT$. Since practical RIS elements possess non-zero resistance, optimized values are then used in~(\ref{prob:main}) to evaluate the performance under the true value of $R$. 
\item BCD: The authors of~\cite{DBLP:IRS_MIMO_SR1} analyzed RIS-assisted MIMO systems for secrecy rate maximization under the assumption of an ideal RIS, where the optimal values of the phase shifts were designed using the BCD technique. Similar to the Ideal PGM case, this benchmark in the present study uses the obtained phase shifts to evaluate the performance under the practical RIS model.

    \item Random-RIS: In this benchmark, Alice adopts the proposed precoding design method, while the RIS operates with random phase shifts.
    \item No-RIS: This benchmark represents communication in the absence of the RIS. Comparing the performance of the proposed method with this benchmark highlights the advantages of RIS integration in PLS~\cite{DBLP:journals/iotj/ChaiBBSN23}.
    
\end{enumerate}
\begin{figure}[t!]
\centering
\includegraphics[ width =0.8\linewidth]{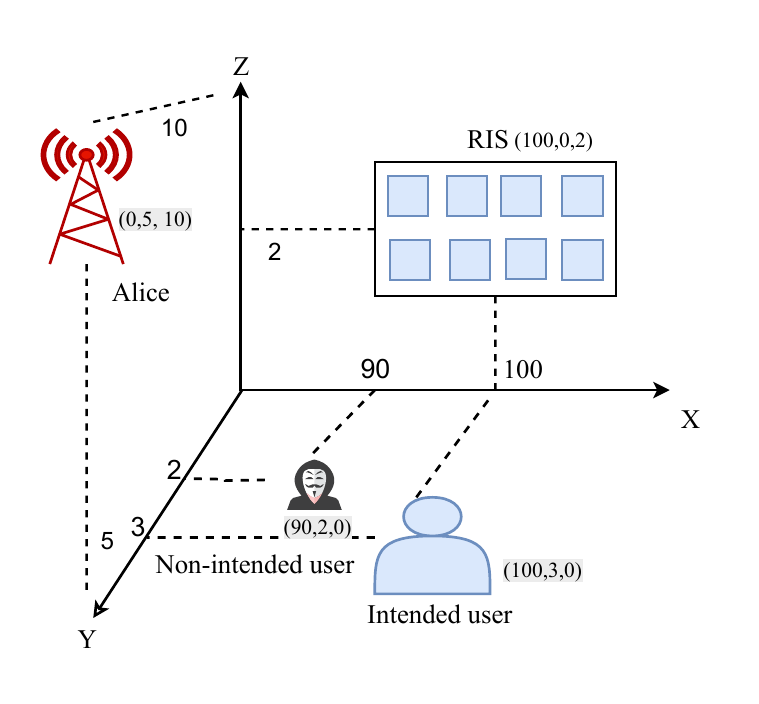}
		 \caption{{ Position of blocks.}}
			\label{fig:postion}
	\end{figure}

\subsection{Secrecy Rate Versus BS Transmit Power}

\begin{figure*}[t!]

 \begin{subfigure}{0.33\textwidth}
 \centering
 \includegraphics[width=0.95\linewidth]{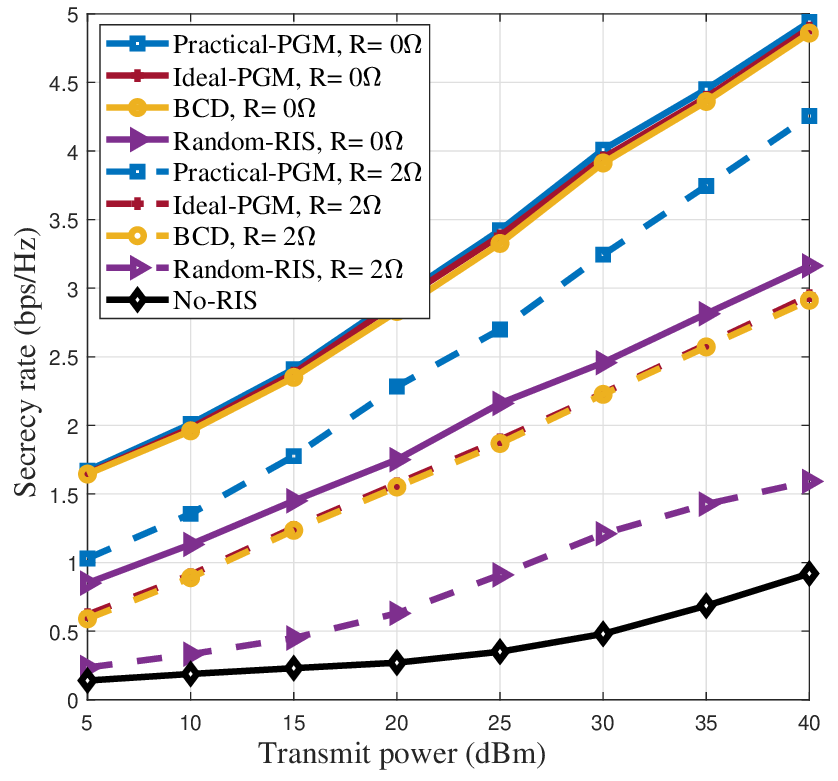}
 \subcaption{} \label{fig:RIS_transmitpower_ideal}
 \end{subfigure}
 \begin{subfigure}{0.33\textwidth}
 \centering
 \includegraphics[width=0.95\linewidth]{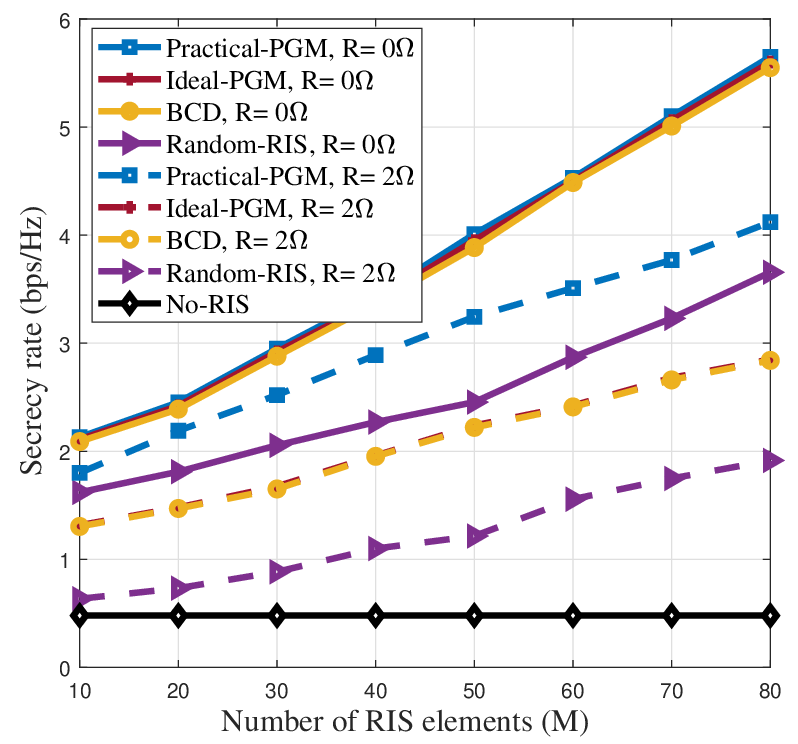}
 \subcaption{} \label{fig:rate_ris_ideal}
 \end{subfigure}
 \begin{subfigure}{0.34\textwidth}
 \centering
\includegraphics[width=1.0\linewidth]{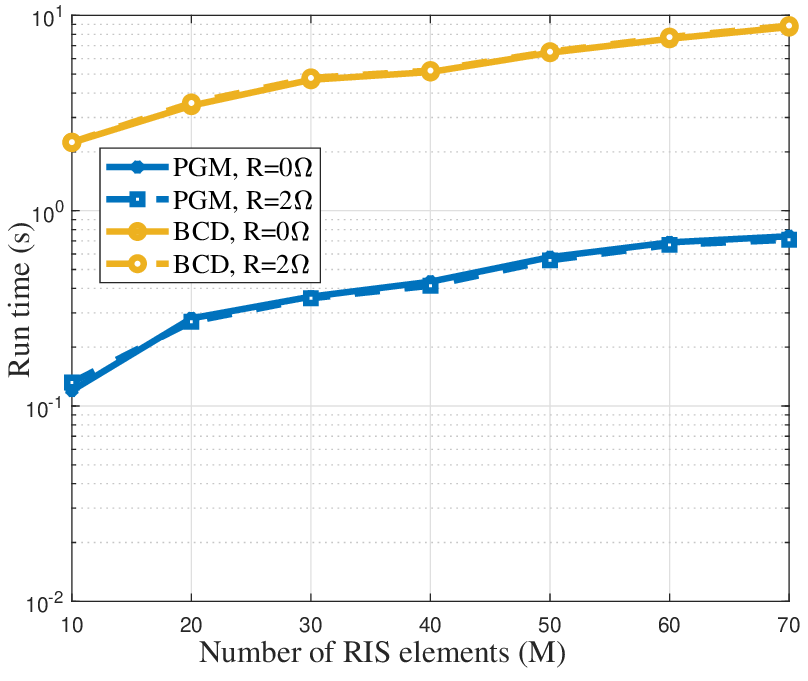}
 \subcaption{} \label{fig:runtime_ris}
 \end{subfigure}

	\hfil
	\vspace{-0.4cm}
	\caption{\footnotesize a) Secrecy rate versus BS transmit power, at $M=50$ (b) Secrecy rate versus number of RIS elements ($M$) at $P=30$ dBm (c) { Run time versus number of RIS elements ($M$), considered value of REs resistance in RIS have ( $R=0\Omega$, $R=2\Omega$), and  $P=30$ dBm.}}
	\label{fig: Rate }
\end{figure*}


Fig.~\ref{fig:RIS_transmitpower_ideal} shows the secrecy rate versus the transmit power at Alice for different RIS configurations and optimization schemes for the cases of an ideal RIS $(R=0\, \Omega)$ and non-ideal $(R=2\, \Omega)$, respectively. All the RIS-assisted schemes significantly outperform the no-RIS benchmark, thereby demonstrating the effectiveness of RIS in enhancing the PLS. When the deployed RIS is ideal, the secrecy rate achieved by the Practical-PGM method coincides with those of the Ideal-PGM and BCD approaches, indicating that all three methods have the same performance. However, when the deployed RIS is non-ideal ($R=2\, \Omega$), the Practical-PGM scheme clearly outperforms both Ideal-PGM and BCD, thereby confirming the importance of explicitly incorporating the RE resistance into the optimization problem. In addition, the secrecy rate under the ideal RIS scenario is higher than that in the non-ideal RIS case, as the hardware-induced reflection losses degrade the system performance. The performance gap between the practical and ideal models becomes more pronounced at higher transmit powers. Overall, the results confirm the superiority of the proposed PGM algorithm and highlight the critical impact of the RIS hardware characteristics on the secrecy performance.

\subsection{Secrecy Rate Versus Number of RIS elements}

Fig.~\ref{fig:rate_ris_ideal} shows the secrecy rate versus the number of RIS elements $(M)$ under the two RIS conditions (ideal $(R=0\, \Omega)$ and non-ideal $(R=2\, \Omega)$). As expected, the secrecy rate monotonically increases with $M$ across all the considered schemes, since large RIS arrays provide higher beamforming gains. The secrecy rate achieved by the practical-PGM coincides with those of the Ideal-PGM and BCD methods when the deployed RIS is ideal. However, when the deployed RIS is non-ideal, the practical-PGM scheme outperforms the ideal-PGM and BCD methods, highlighting its robustness in the presence of hardware constraints. The secrecy rate in the ideal-RIS scenario shows a better performance than the practical RIS method due to the absence of hardware-induced losses, and the performance gap between the two methods widens as $M$ increases. Finally, the no-RIS scheme exhibits a flat secrecy rate with respect to $M$ and performs substantially worse than all the RIS-assisted schemes, thereby underscoring the crucial role of RIS in enhancing PLS.

\subsection{Run Time Versus Number of RIS Elements}

Fig.~\ref{fig:runtime_ris} shows the run time performance of proposed  PGM Algorithm~\ref{alg:PGM} and BCD  as a function of the number of RIS elements $(M)$ for $R=0 \, \Omega$ and $R=2 \, \Omega$. The proposed PGM algorithm exhibits a significantly lower run time compared to BCD across all values of $M$. In particular, for all values of $M$ under investigation, PGM is approximately one order of magnitude faster than BCD. This phenomenon can be explained by the difference in the computational complexities of the two methods, as indicated in Table \ref{tab:complexity_all}. As observed, the per-iteration computational complexity of the proposed PGM method is linear with the number of REs ($\mathcal{O}\left(M\right)$). In contrast, for the BCD method, it is quadratic ($\mathcal{O}\left(M^2\right)$). In addition, for both PGM and BCD, the run times are almost identical under $R=0 \, \Omega$ and $R=2 \, \Omega$, which indicates that the runtimes of both methods are insensitive to the RE resistance value. Overall, the results demonstrate that the proposed PGM algorithm not only achieves significant reductions in the run time compared with the BCD method but also offers excellent scalability. Consequently, it represents a promising approach for real-time optimization of the RIS phase shifts in large-scale RIS-assisted communication systems.

\begin{table}[t!]
\renewcommand{\arraystretch}{1.25}
    \centering
    \caption{ Overall computational complexity}
    \label{tab:complexity_all}
    \begin{tabular}{|p{0.38\linewidth}|p{0.628\linewidth}|}
        \hline
        \textbf{Method}  & \textbf{Complexity} \\\hline         
     Proposed \textbf{Algorithm~\ref{alg:PGM}} PGM & $\mathcal{O} \left( I_1 \left( M \Nx^2 + \Nx^3 + I_{2 \theta} (M \Nx^2 + \Nx^3) + I_{2 \text{t}} \Nx^3 \right) \right)$ \\ 
  \hline
  BCD &  $\mathcal{O} \left( L_1 \left( 2 \Nx^2 + 3\Nx^3 +M^2 + T \Nx^3 + (\Nx^3 +\Nx^2)  \right) \right)$\\
  \hline
    \end{tabular}
\end{table}


\begin{figure*}[t!]

 \begin{subfigure}{0.33\textwidth}
 \centering
 \includegraphics[width=0.95\linewidth]{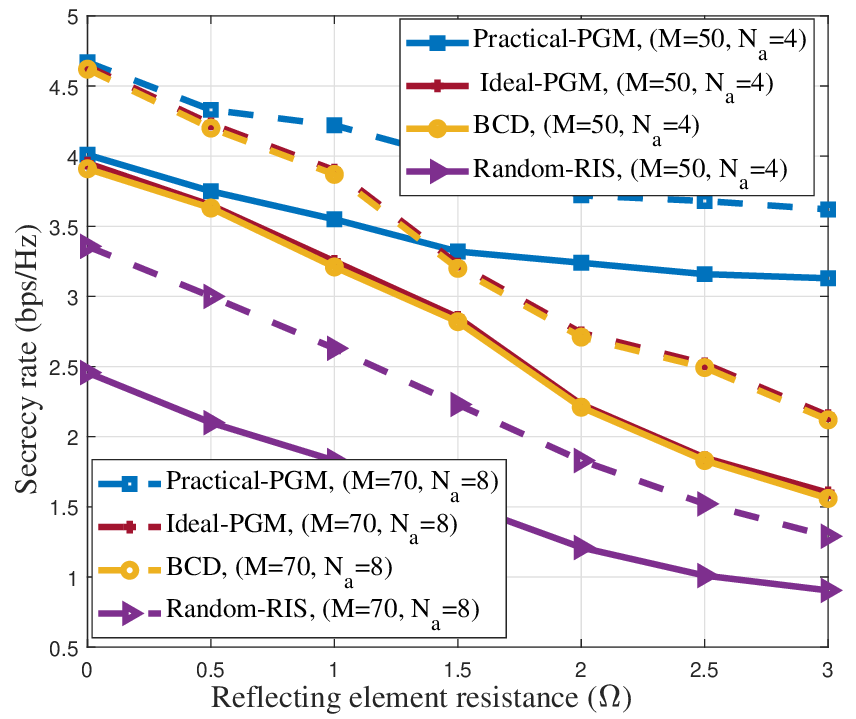}
 \subcaption{} \label{fig:rate_risvalue}
 \end{subfigure}
 \begin{subfigure}{0.33\textwidth}
 \centering
 \includegraphics[width=0.95\linewidth]{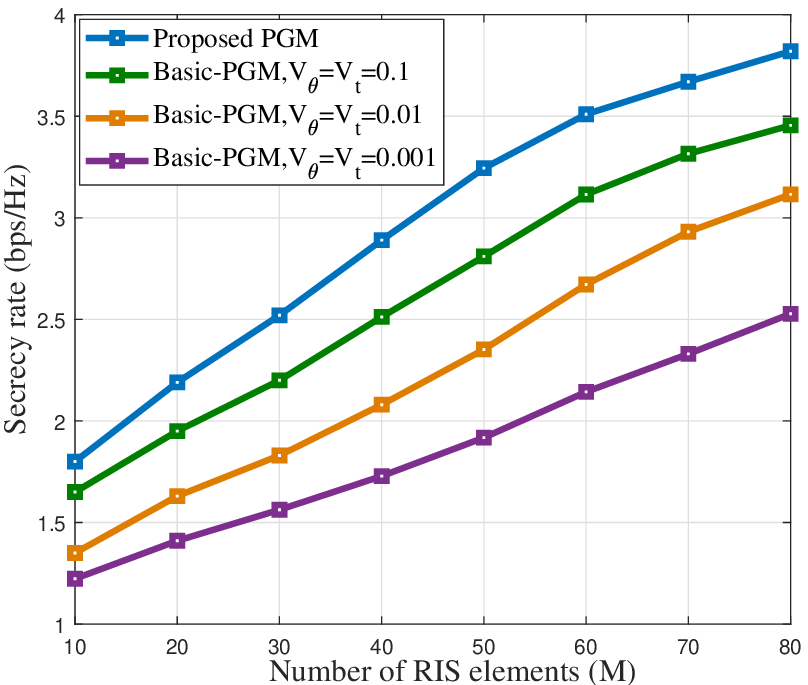}
 \subcaption{} \label{fig:rate_ris_stepsize} 
 \end{subfigure}
 \begin{subfigure}{0.32\textwidth}
 \centering
\includegraphics[width=0.95\linewidth]{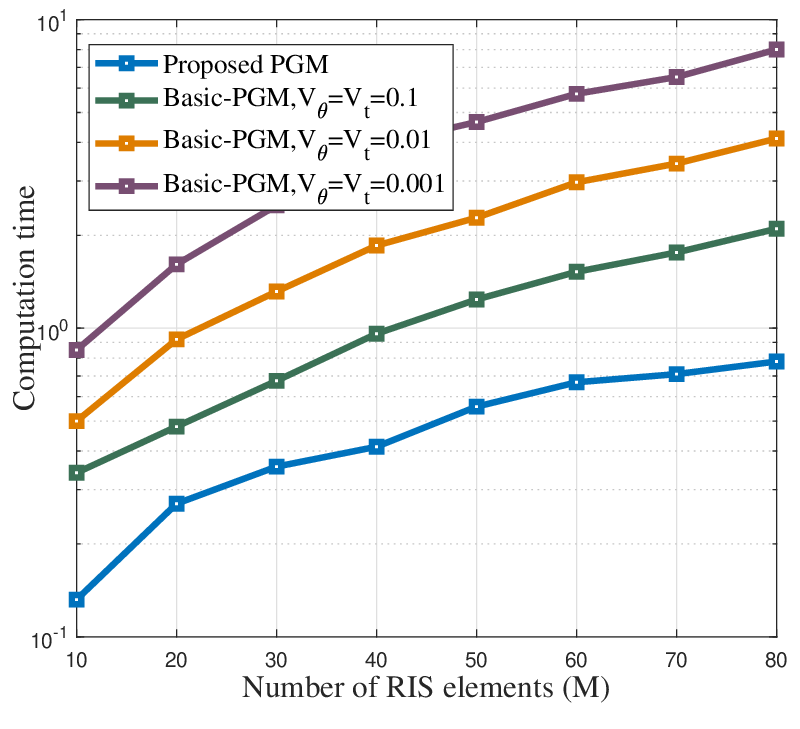}
 \subcaption{} \label{fig:compu_ris_stepsize}
 \end{subfigure}

	\hfil
	\vspace{-0.4cm}
	\caption{\footnotesize a) { Secrecy rate versus value of REs resistance ($\Omega$), number of RIS $M=\{50,70\}$, number of antenna at Alice $\Na =\{4,8\}$ and $P=30$ dBm} (b) { Secrecy rate versus RIS elements $(M)$ for different value of step size with $R=2\Omega, P=30\text{dBm}$} (c)  Computational time versus number of RIS elements $(M)$ with $R=2\Omega, P=30\text{dBm}$. }
	\label{fig: Rate_1 }
\end{figure*}


\subsection{Secrecy Rate Versus Reflecting Element Resistance} 

Fig.~\ref{fig:rate_risvalue} illustrates the variation of the secrecy rate with respect to the RE resistance under the considered RIS models for two scenarios: i) $M=50$, $\Na = 4$, and ii) $M=70$, $\Na=8$. For both models, and both scenarios, the secrecy rate monotonically decreases with increasing RE resistance due to the greater resistive losses at the RIS, which reduce the reflection gains. Among the considered methods, the proposed Practical-PGM algorithm consistently achieves the highest secrecy rate and exhibits a relatively slower performance degradation as the resistance increases. It thus demonstrates strong robustness against hardware constraints. On the other hand, the Ideal-PGM and BCD schemes exhibit more severe performance degradations as the RE resistance increases, since neither method accounts for the RE resistive losses in the optimization process. Among all the considered methods, the Random-RIS scheme shows the poorest performance, with the secrecy rate dropping sharply under high RE resistance values, highlighting the inefficiency of unoptimized phase configurations. Increasing the number of RIS elements (from $M=50$ to $ = M=70$) and transmit antennas (from $ N_a=4$ to $ N_a=8$) enhances the secrecy performance across all the schemes due to the stronger beamforming gain. However, the impact of the RE resistance is more pronounced in larger RIS arrays, underscoring the need to account for hardware non-idealities in the design of RIS-assisted secure communication systems.

\subsection{Effects of Step-Size Control and Initialization on PGM }

This section analyzes the impact of the step size parameters, $\upsilon_{\theta}$ and $\upsilon_{t}$, on the performance of the PGM method. Four scenarios are considered: i) the proposed PGM method, which performs step size initialization according to $\upsilon_{\theta}$ and   $\upsilon_{t}$ along with the adaptive control mechanism discussed in Section~\ref{sec:join opt}, ii) basic PGM with fixed $\upsilon_{\theta} =\upsilon_{t}=0.1$, iii) basic PGM with fixed $\upsilon_{\theta} =\upsilon_{t}=0.01$, and iv) basic PGM with fixed $\upsilon_{\theta} =\upsilon_{t}=0.001$. The variations of the secrecy rate and computational time with the number of RIS elements $(M)$ with $R=2\Omega$ are shown in Fig.~\ref{fig:rate_ris_stepsize} and Fig.~\ref{fig:compu_ris_stepsize}, respectively. Owing to its adaptive step-size initialization and control strategy, the proposed PGM method achieves higher secrecy rates and substantially lower computational times than the basic PGM schemes with fixed step sizes. Furthermore, as the fixed step-size decreases, the performance gap between the proposed PGM method and the three PGM variants widens, indicating that a small fixed step size not only degrades the secrecy rate but also increases the computational cost due to slower convergence. Overall, the results confirm that proper step-size initialization combined with adaptive control is crucial for simultaneously enhancing the secrecy performance and reducing the convergence time in RIS-assisted systems.

In addition to addressing the secrecy rate maximization problem, this study also investigates the suboptimal channel power difference maximization (CPDM) formulation presented in~(\ref {prob:CPDM}). In this approach, the optimal values of the precoder are defined using eigen decomposition, as explained in Sec. \ref{sec:CPDM}, and the phase shifts are obtained using Algorithm~\ref{alg:CPDM}. To confirm the efficacy of the proposed method,  the following benchmark techniques were employed:

\begin{itemize}
    \item  PGM: This scenario considers the proposed method, where $\TT^{\text{(init)}}$ is obtained via eigen decomposition, as described in Sec. \ref{sec:CPDM}, and $\thetaa^{\text{(init)}}$ is obtained using Alg.~\ref{alg:CPDM}.
    
    \item DSM: The authors in \cite{DBLP:journals/wcl/SirojuddinPH22} proposed a method called Dimension-wise Sinusoidal Maximization (DSM) to optimize the phase shifts $\thetaa$ of a RIS-aided MIMO system incorporating the ideal RIS model. In DSM, the precoder $\TT$ is obtained using SVD given fixed phase shifts.
    \item SDR: Cui et al.~\cite{DBLP:journals/wcl/CuiZZ19} proposed a semidefinite relaxation method (SDR) followed by Gaussian randomization to design the phase shifts of the RIS by transforming problem (\ref{prob:CPDM}) to a semidefinite problem (SDP) and then using the SDR method to solve it. This method assumes the ideal RIS model. In addition, the precoding matrix $\TT$ is obtained using eigen decomposition, similar to the approach adopted in the present study.
    \item RND: In this scenario, when solving the problem in~(\ref{prob:CPDM}), $\thetaa$ and $\TT$ are randomly selected while satisfying $\trace \left( \TT \TT^H\right) = P$ and $\thetamin \leq \theta_m \leq \thetamax, \forall m$.
\end{itemize}


\begin{figure*}[t!]

 \begin{subfigure}{0.33\textwidth}
 \centering
 \includegraphics[width=0.95\linewidth]{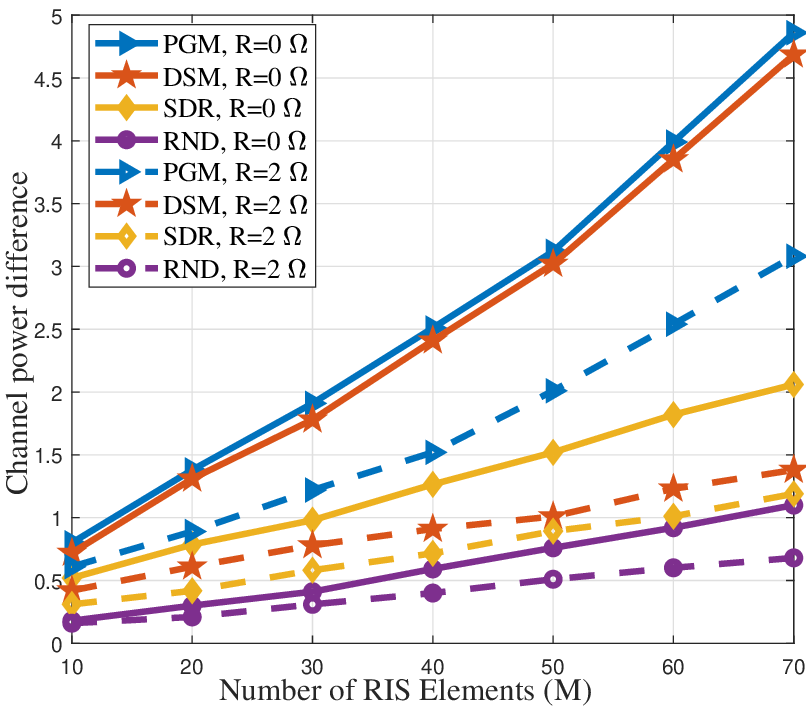}
 \subcaption{} \label{fig:cpdm_ris}
 \end{subfigure}
 \begin{subfigure}{0.35\textwidth}
 \centering
 \includegraphics[width=0.98\linewidth]{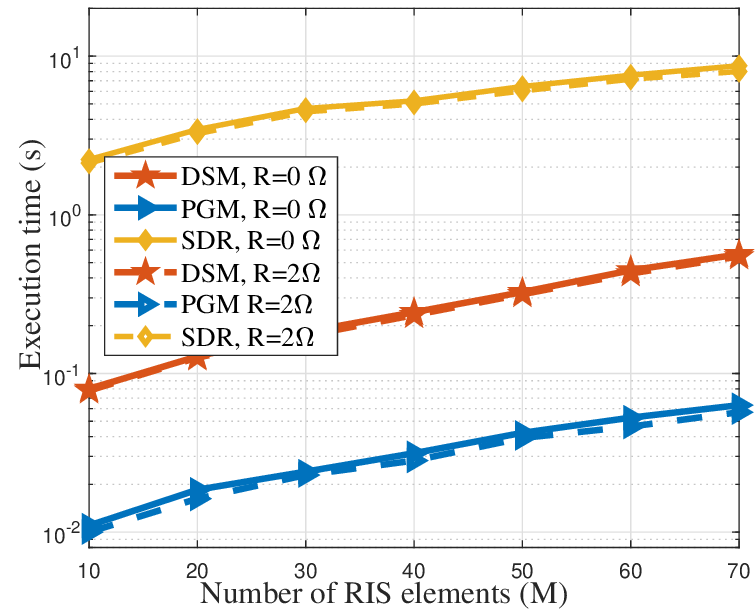}
 \subcaption{} \label{fig:cpdm_runtime}
 \end{subfigure}
 \begin{subfigure}{0.33\textwidth}
 \centering
\includegraphics[width=0.98\linewidth]{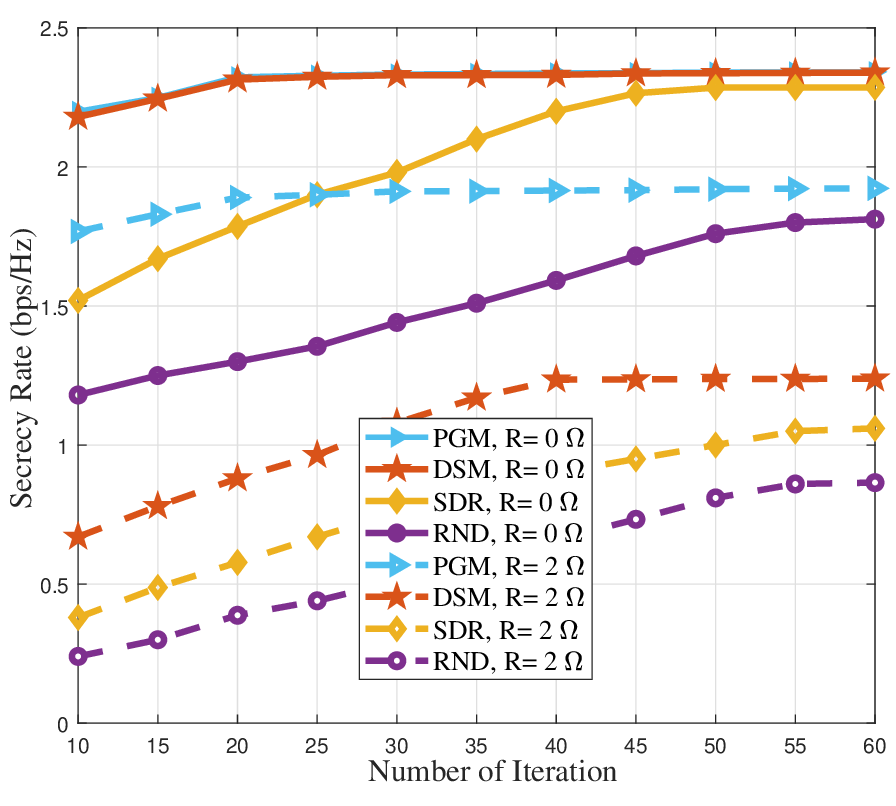}
 \subcaption{} \label{fig:rate_itt}
 \end{subfigure}

	\hfil
	\vspace{-0.4cm}
	\caption{\footnotesize a) { Channel power difference versus number of RIS elements ($M$) for two values of RE: $R=0\Omega$, and $R=2\Omega$ } (b) { Execution time versus number of RIS elements ($M$) for values of RE resistance: $R=0\Omega$, and $R=2\Omega$ } (c) { Convergence secrecy rate to update theta ($\theta$) in Algorithm.~\ref{alg:PGM} under different initials point and two considered values of RE: $R=0\Omega$, and $R=2\Omega$, $M=50$, $P=30$ dBm}.}
	\label{fig: Rate_2 }
\end{figure*}

\begin{table}[t!]
\renewcommand{\arraystretch}{1.25}
    \centering
    \caption{ Overall computational complexity}
      \label{tab:complexity_all1}
   \begin{tabular}{|p{0.31\linewidth}|p{0.58\linewidth}|}
        \hline
        \textbf{Method}  & \textbf{Complexity} \\\hline         
   Proposed \textbf{Algorithm~\ref{alg:CPDM}} & $\mathcal{O} \left( \bar{I}_1 \left( M \Nx^2 + \bar{I}_2 M \Nx^2 \right) \right)$ \\ 
  \hline
  DSM &  $\mathcal{O}  \left( 2 \Nx^2 + \Nx^2 ( \Nx +M) + I_D M(M-1)  \right) $\\
  \hline
   SDR &  $\mathcal{O} \left( I_S \left( \Nx^3 + (M+1)^{3.5}  \right) \right)$\\
  \hline
    \end{tabular}
\end{table}

Fig.~\ref{fig:cpdm_ris} shows the variation of the channel power difference  with the number of RIS elements $(M)$, considering two resistance values of the REs, $R=0 \, \Omega$ and $R=2 \, \Omega$. In the ideal RIS case, the PGM and DSM methods show a similar performance, which aligns with the findings discussed in~\cite{DBLP:IRS-MIMO_capacity_siro}. The SDR approach, however, exhibits a lower channel power difference due to its indirect optimization process that involves relaxing the rank constraint of the matrix in the SDP step. Since the resulting solution does not guarantee a rank-one outcome, an additional Gaussian randomization step is required, which introduces non-stationarity and limits the potential of the SDR method. In the non-ideal RIS scenario, the proposed Algorithm ~\ref{alg:CPDM} outperforms DSM and SDR, highlighting the importance of accounting for the resistance of the REs when designing a sub-optimal phase shift solution to the optimization problem~(\ref{prob:CPDM}). When the resistance increases from $R=0 \, \Omega$ to $R=2\Omega$, all of the schemes exhibit a lower channel power difference, as expected, due to the resistive losses caused by the non-ideal RIS hardware. Finally, the RND method consistently yields the lowest channel power difference for both $R=0 \, \Omega$ and $R=2 \, \Omega$, reinforcing the need for an optimization-based phase shift design to solve the problem~(\ref{prob:CPDM}).

Fig.~\ref{fig:cpdm_runtime} shows the execution times of the PGM, DSM, and SDR methods as a function of the number of RIS elements $(M)$ under two resistance values ($R= \{0,2\}$). PGM provides the fastest computational time of the three methods. However, DSM is faster than SDR. Referring to Table \ref{tab:complexity_all1}, the SDR method incurs the highest per-iteration computational cost, i.e., $\mathcal{O}(M)^{3.5}$; therefore, its execution time increases rapidly as $M$ grows. In contrast, the per-iteration computational costs of the PGM and DSM methods are $\mathcal{O}(M)$ and $\mathcal{O}(M)^2$, respectively. Notably, the execution times for all of the methods are nearly identical for both $R=0 \, \Omega$ and $R=2 \, \Omega$, which indicates that the computational complexities of the three methods are independent of the resistance of the REs. 

In many studies, a sub-optimal criterion is used to generate an initial search point for the main algorithm in an attempt to accelerate the convergence behavior. In this study, the sub-optimal values of $\TT$ and $\thetaa$, obtained by solving the optimization problem~(\ref{prob:CPDM}), were employed as $\TT^{\text{(init)}}$ and $\thetaa^{\text{(init)}}$ in Algorithm~\ref{alg:PGM}, respectively. To explore the relationship between the CPDM method and the secrecy rate maximization objective, Fig. \ref{fig:rate_itt} shows the convergence behavior of Algorithm \ref{alg:PGM} using initial values of $\TT$ and $\thetaa$ obtained from four different methods (PGM, DSM, SDR, and RND) under two scenarios: the ideal RIS case with $R=0 \, \Omega$ and the non-ideal RIS case with $R=2 \, \Omega$. The results show that PGM and DSM yield higher convergence rates than SDR and RND for both values of $R$. In the ideal RIS case, PGM and DSM converge in $20-25$ iterations, whereas SDR and RND converge in $45-50$ and $60$ iterations, respectively. However, when $R=2\, \Omega$, PGM still shows a better starting point and converges in $15-20$ iterations, but DSM converges only after approximately $40-45$ iterations. Hence, PGM not only provides an accurate initial estimate for $\TT$ and $\thetaa$ in Algorithm ~\ref{alg:PGM}, but also exhibits improved robustness under variations in the RE values of the RIS.

\section{Conclusion}
\label{conclusion}

This paper investigated the secrecy rate maximization problem in a practical RIS-assisted MIMO system where the REs possess finite resistance values, which reduce the reflection gains and hence degrade the overall system performance. The inclusion of resistance in the system model couples the design variables and makes the joint optimization of the precoding matrix and RIS phase shifts highly non-convex. To address this challenge, this study proposed a low-complexity PGM framework to efficiently solve the secrecy rate maximization problem. An adaptive step-size initialization and control mechanism were further incorporated into the PGM to enhance the convergence stability of the iteration process while also reducing the computational burden compared to conventional fixed step-size approaches. Additionally, a channel power difference maximization (CPDM) sub-problem was formulated to enable the efficient design of the RIS phase shifts with a lower complexity, making the approach particularly attractive for large-scale RIS-assisted MIMO networks.

Comprehensive simulation results verified the effectiveness of the proposed PGM scheme. Specifically, the PGM-based optimization scheme consistently outperformed conventional algorithms such as BCD and random RIS under both ideal and practical RIS models, with the performance gains being particularly significant in practical settings involving non-zero RE resistance. Moreover, the CPDM-based design achieved an effective trade-off between the secrecy performance and the computational complexity, highlighting its suitability for real-time implementation in large-scale systems.
Overall, the methods and results presented in this study provide new insights into secure transmission design with practical RIS hardware and establish scalable optimization frameworks that can be extended to other advanced communication scenarios, such as RIS-aided ISAC and multi-user secure networks.
In future studies, the proposed framework can be extended to other practical scenarios, including multi-RIS cooperative deployments, mobile user environments, and optimization under imperfect CSI, thereby further enhancing the robustness and applicability of secure RIS-assisted communications in beyond-5G and 6G networks.

\appendices
\section{Derivation of (\ref{eq:dCsec dthetaa})} \label{sec:derivations}
Since $\Csec = \Rb - \Re$, then  $\partial \Csec / \partial \thetaa = \partial \Rb / \partial \thetaa - \partial \Re / \partial \thetaa$. We first compute $\partial \Rb / \partial \thetaa$ in detail. The derivation process for $\partial \Re / \partial \thetaa$ then follows an analogous approach. To begin, we evaluate the scalar $\partial \Rb / \partial \theta_m$ and then generalize the result to obtain the vector form $\partial \Rb / \partial \thetaa$. From (\ref{eq:Cb}), since
\begin{equation*}
    \Rb = \log_2 \det \left( \II_{\Nb} + \Hhatb \TT \TT^H \Hhatb^H / \sigmab \right),
\end{equation*}
with $\Hhatb = \Hab + \Hrb \Phii \Har$, by using the chain rule for derivation, it is easily shown that \begin{align*}
    \frac{\partial \Rb}{\partial \theta_m} &= \frac{1}{\ln 2} \left[ \AAb^{-1} \frac{\partial \AAb}{\partial \theta_m} \right] \\
    &= \frac{1}{\ln 2} \tr \left[ \AAb^{-1} \frac{1}{\sigmab} \left( \frac{\partial \Hhatb}{\partial \theta_m} \TT \TT^H \Hhatb^H + \Hhatb \TT \TT^H \frac{\partial \Hhatb^H}{\partial \theta_m} \right) \right]
\end{align*}
with
\begin{equation*}
    \frac{\partial \Hhatb}{\partial \theta_m} = \Hrb \frac{\partial \Phii}{\partial \theta_m} \Har = \frac{\partial \phi_m}{\partial \theta_m} \Hrb \ee_m \ee_m^T \Har,
\end{equation*}
where $\ee_m \in \mathbb{R}^{M}$ is the standard basis vector whose $m$th element is one and the remaining elements are zero. Since $\left( \partial \Hhatb / \partial \theta_m \right) \TT \TT^H \Hhatb^H$ and $\Hhatb \TT \TT^H \left( \partial \Hhatb^H / \partial \theta_m \right)$ are the conjugate transposes of each other, it follows that
\begin{align*}
    \frac{\partial \Rb}{\partial \theta_m} &= \frac{2}{\ln 2} \Ree \left\{ 
    \frac{\partial \phi_m}{\partial \theta_m} \frac{1}{\sigmab} \tr \left[ \AAb^{-1} \Hrb \ee_m \ee_m^T \Har \TT \TT^H \Hhatb^H \right] \right\} \\
    & \eqtext{(1)} \frac{2}{\ln 2} \Ree \left\{ \frac{\partial \phi_m}{\partial \theta_m} \frac{1}{\sigmab} \ee_m^T \Har \TT \TT^H \Hhatb^H \AAb^{-1} \Hrb \ee_m \right\} \\
    & \eqtext{(2)} \frac{2}{\ln 2} \Ree \left\{ \frac{\partial \phi_m}{\partial \theta_m} \frac{1}{\sigmab} \left[ \Har \TT \TT^H \Hhatb^H \AAb^{-1} \Hrb \right]_{m,m} \right\},
\end{align*}
where equality (1) exploits the fact that the trace operation is invariant under cyclic permutation, and equality (2) is valid since $\ee_m^T \XX \ee_n$ means taking an element of matrix $\XX$ in row $m$ column $n$. The value of vector $\partial \Rb / \partial \thetaa$ is then given as
\begin{equation*}
    \frac{\partial \Rb}{\partial \thetaa} = \frac{2}{\ln 2} \Ree \left\{ \frac{\partial \phii}{\partial \thetaa} \circ \diag \left( \frac{1}{\sigmab} \Har \TT \TT^H \Hhatb^H \AAb^{-1} \Hrb \right) \right\}.
\end{equation*}
Adopting a similar approach, the value of $\partial \Re / \partial \thetaa$ is expressed as
\begin{equation*}
    \frac{\partial \Re}{\partial \thetaa} = \frac{2}{\ln 2} \Ree \left\{ \frac{\partial \phii}{\partial \thetaa} \circ \diag \left( \frac{1}{\sigmae} \Har \TT \TT^H \Hhate^H \AAe^{-1} \Hre \right) \right\}.
\end{equation*}
By combining $\partial \Rb / \partial \thetaa$ and $\partial \Re / \partial \thetaa$, the result in (\ref{eq:dCsec dthetaa}) is obtained.

\section{Proof of Proposition \ref{prop:CPDM}} \label{sec:CPDM proof}
Let us denote the eigen-decompositions of $\FFb$ and $\FFe$ as $\FFb = \UUb \BB \UUb^H$ and $\FFe = \UUe \EE \UUe^H$, respectively. It then follows that
\begin{align*}
    \Csec &= \log_2 \det \left( \II_{\Ns} + \BB \right) - \log_2 \det \left( \II_{\Ns} + \EE \right) \\
    &= \frac{1}{\ln 2} \left\{ \sum_{n=1}^{\Ns} \left[ \ln \left( 1 + b_n \right) - \ln \left( 1 + \epsilon_{n} \right) \right] \right\},
\end{align*}
where $b_n \geq 0$ and $\epsilon_n \geq 0$, $\forall n \in \mathcal{N}_{\text{s}} = \{1, \dots, \Ns \}$, are the eigenvalues of $\BB$ and $\EE$ arranged in non-decreasing order, respectively.
By using the inequality $\ln \left( 1 + x_n \right) \leq x_n$ for $x_n \geq 0$, it is readily obtained that
\begin{align*}
    \Csec & \leqtext{(1)} \frac{1}{\ln 2} \sum_{n=1}^{\Ns} \left( b_n - \epsilon_n \right) = \frac{1}{\ln 2} \trace \left( \FFb - \FFe \right) = \frac{1}{\ln 2} \Pdiff,
\end{align*}
where inequality (1) is valid due to the following two facts. First, since $\FFb - \FFe$ is positive semidefinite (PSD), then $b_n \geq \epsilon_n$ based on the Courant-Fischer theorem \cite{Horn90Matrix}. Second, the inequality $\ln \left(1 + b_n \right) - \ln \left( 1 + \epsilon_n \right) \leq b_n - \epsilon_n$, or equivalently, $b_n - \ln \left( 1 + b_n \right) \geq \epsilon_n - \ln \left( 1 + \epsilon_n \right)$, is true by proving that $\zeta (x) = x - \ln (1+x)$ is increasing for $x \geq 0$. It is obvious that $\zeta'(x) = 1 - 1/(1+x) = x/(x+1) \geq 0$ is non-negative, so $\zeta(x)$ in increasing for $x \geq 0$. This proves the first sentence in Proposition \ref{prop:CPDM}.

By formulating the Taylor expansion of the log-determinant of a PSD matrix $\XX$ as $\ln \det (\II + \XX) = \trace (\XX) - (1/2) \trace (\XX^2) + (1/3) \trace (\XX^3) + \dots$, it is readily shown that
\begin{align*}
    \Csec = \frac{1}{\ln 2} \bigg[ \trace \left( \FFb - \FFe \right) - \frac{1}{2}  \trace \left( \FFb^2 - \FFe^2 \right) + \dots \bigg].
\end{align*}
If $b_n$ and $\epsilon_n$, $\forall n$, are small, the linear terms dominate and the higher-order (negative) corrections are small. Therefore,
\begin{align*}
    \Csec \approx \frac{1}{\ln 2} \trace \left( \FFb - \FFe \right).
\end{align*}
This proves the second sentence in Proposition \ref{prop:CPDM}.

\bibliographystyle{IEEEtran}

\bibliography{IEEEabrv,IRS_MIMO_references}

\end {document}